\documentclass{amsart}
\usepackage[colorlinks,citecolor=blue]{hyperref}
\usepackage{cases}
\usepackage{graphicx}
\usepackage[dvipsnames, svgnames, x11names]{xcolor}

\newtheorem{theorem}{Theorem}[section]
\newtheorem{lemma}[theorem]{Lemma}

\theoremstyle{definition}
\newtheorem{definition}[theorem]{Definition}
\newtheorem{example}[theorem]{Example}

\newtheorem{assumption}[theorem]{Assumption}
\newtheorem{algorithm}[theorem]{Algorithm}

\theoremstyle{remark}
\newtheorem{remark}[theorem]{Remark}

\numberwithin{equation}{section}

\def\N{{\mathbb N}}

\def\R{{\mathbb R}}

\def\F{{\mathbb {F}}}

\def\X{{\mathbb{X}}}

\def\HK{\hbox{\rm{HK}}}

\def\res{\hbox{\rm{res}}}
\def\LC{\hbox{\rm{LC}}}

\def\MoCont{\hbox{\rm{MoCont}}}
\def\MoPrim{\hbox{\rm{MoPrim}}}
\def\GCD{\hbox{\rm{GCD}}}
\def\s{{\mathbf{s}}}
\def\PC{{\mathcal{C}}}
\def\DC{{\mathcal{D}}}
\def\argmax{{\hbox{\rm{argmax}}}}

\begin{document}

\title{Bit Complexity of Polynomial GCD on Sparse Representation}

\author{Qiao-Long Huang}
\address{School of Mathematics, Shandong University}
\email{huangqiaolong@sdu.edu.cn}
\thanks{The first author was supported in part by NSFC Grant No.12001321.}

\author{Xiao-Shan Gao}
\address{UCAS, Academy of Mathematics and Systems Science,
 Chinese Academy of Sciences}
\email{xgao@mmrc.iss.ac.cn}
\thanks{The second author was supported in part by NSFC Grant No. 12288201.}





\begin{abstract}
An input- and output-sensitive GCD algorithm for multi-variate polynomials over finite fields is proposed by combining the modular method with the Ben-Or/Tiwari sparse interpolation. 
%
The bit complexity of the algorithm is given and is sensitive to the sparse representation, while for previous sparse GCD algorithms, the complexities were given only in some special cases.
It is shown that the new algorithm is superior both in theory and in practice
comparing with existing GCD algorithms: the complexity in the degree
is decreased from quadratic to linear and the running times
are decreased by 1-3 orders of magnitude in various benchmarks.
\end{abstract}

\maketitle

\section{Introduction}

%


Multivariate polynomial GCD computation is one of the central problems in algebraic
and symbolic computation.
In 1967, Collins \cite{Collins67} gave the first major advance by proposing
a refined Euclidean style algorithm.
Such direct computational algorithms lack scalability due to the so-called intermediate expression swell phenomenon.
The effective idea to solve the intermediate expression swell problem is modular algorithms,
that is, by substituting some of the variables by certain integers,
multivariate polynomial GCD computation becomes univariate GCD computation over finite fields,
and the true GCD will be recovered from these univariate GCDs either by
interpolations or by the Chinese Remainder Theorem.
In 1971, Brown \cite{brown1971euclid} gave the first modular algorithm based on interpolation for dense polynomials.

The above algorithms are for dense polynomials.
In 1973, Moses and Yun proposed the EZ-GCD algorithm \cite{MosesY73}, where  Hensel lifting instead of interpolation was used to recover the GCD.
%
%
%
In 1980, Wang \cite{Wang80b} proposed an enhanced EZ-GCD algorithm, called EEZ-GCD algorithm, which improved the EZ-GCD algorithm by solving the leading coefficient problem, bad-zero problem, unlcuky evaluation problem, and the common divisor problem.

In 1979, Zippel \cite{Zippel79} developed the first modular sparse GCD algorithm based on sparse polynomial interpolations, which interpolates the GCD one variable at a time.
Zippel's algorithm is probabilistic and its correctness relies on the  Schwartz-Zippel lemma.
In 1988, Kaltofen \cite{kaltofen1988greatest} gave a GCD algorithm for polynomials given by straight-line programs.
%
In 1990, Kaltofen and Trager \cite{kaltofen1990computing} gave a GCD algorithm for polynomials given by black boxes.
In 2008, Cuyt and Lee \cite{cuyt2008new} proposed another improved technique. For one evaluation, their algorithm reduces a multivariate polynomial into a univariate polynomial.
In 2016, Hu and Monagan \cite{hu2016fast} presented a parallel GCD algorithm for sparse GCD computation, which combined a Kronecker substitution with a Ben-Or/Tiwari sparse interpolation~\cite{Ben-OrT88} modulo a smooth prime to determine the support of the GCD.
In 2018, Tang, Li, and Zeng \cite{tang2018computing} proposed
two methods based on variations of Zippel's method and Ben-Or/Tiwari's interpolation algorithm ~\cite{Ben-OrT88}, respectively.

Despite of the vast literatures, explicit bit complexities for sparse multivariate polynomial GCD algorithms
seems not given. In previous work, the bit complexities were either mentioned to be polynomial
in the number of variables, degrees, and the number of terms of the input polynomials
or given under certain conditions.
In this paper, we will give a new GCD algorithm and its bit complexity,
which is sensitive for the sparse representation.
The given algorithm is shown to have better complexities
and much better practical performance than existing algorithms.

\subsection{Main results}
Let $\F_q$ be a finite field with $q$ elements, where $q$ is a prime or prime power.
In this paper, we focus on GCD computation over finite fields. Let $A$ and $B$ be two polynomials in $\F_q[x_1,\dots,x_n]$ and $G=\gcd(A,B)$.
In the following, $T_A,T_B$, and $T_G$   are respectively the numbers of terms in $A,B$ and $G$.
$D$ (and $d$) is the degree (and partial degree) bound of $A$ and $B$.
The algorithm is randomized, so we assume that we can obtain a random bit with bit-cost $O(1)$.
The main result of the paper is given below.

\begin{theorem}
Let $A,B$ be in $\F_q[x_1, \dots,x_n]$, and suppose that a primitive root $\omega$ of $\F_q$ is given.
For any $\varepsilon\in(0,1)$, there exists an algorithm that takes as inputs $A,B$ and
returns $G=\gcd(A,B)$ with probability at least $1-\varepsilon$ using $O^\sim(nDT_G(T_A+T_B)\log^2\frac{1}{\varepsilon}\log^2 q)$ bit operations.
\end{theorem}

Our algorithm may fail to find the correct number $T_G$, which leads to an endless run. At this point, we force quit when the wrong $T_G$ reaches $(d+1)^n$. But luckily, this case only happens with probability $\leq \varepsilon$. So if we choose $\varepsilon$ small enough, for example $\varepsilon=\frac{1}{nD(d+1)^n(T_A+T_B) \log^2 q}$,
then the expected complexity is $O^\sim(n^3DT_G(T_A+T_B)\log^2 q)$ bit operations.

The algorithm is implemented in Maple and extensive numerical experiments
show that the new algorithm outperforms the default GCD in Maple by 1-3 orders of magnitudes
as shown by Table \ref{tab-10}. Details of the experiments can be found in section \ref{sec-exp}.

\begin{table}[ht]
\centering
\begin{tabular}{ccc}
  \hline
%
Experiment Settings & Maple GCD & our GCD\\
 \hline
$n=6$, $D=30$, $t$ varies, $t_l\le 60$s &
$t_m \approx 18$ &
$t_m \approx 150$  \\
$t=30$, $D=100$, $n$ varies, $t_l\le 60$s &
$n_m \approx 3$ &
$n_m \approx 200$  \\
$n=6$, $t=30$, $D$ varies, $t_l\le 100$s &
$D_m \approx 23$ &
$D_m \approx 29525$  \\
\hline
\end{tabular}
\caption{Experimental results for computing $G=\GCD(A,B)$,
where
$D=\deg A=\deg B=\deg G$,
$t=\# A=\# B=\# G$,
$t_l$ is the running time threshold in seconds.
We use $t_m$, $n_m$, $D_m$ to denote the maximum
terms, numbers of variables, degrees that can be computed within the 
given time threshold $t_l$.   }
\label{tab-10}
\end{table}

\vskip-20pt\,
At top level, the algorithm is a combination of the modular method with the Ben-Or/Tiwari sparse interpolation~\cite{Ben-OrT88}.
Main ingredients of the algorithm include:
a new variable substitution is introduced  to isolate the leading coefficient of the GCD, that is, the leading coefficient of substituted GCD is a monomial;
the concept of diverse polynomials introduced by Giesbrecht and Roche \cite{giesbrechtD11}
is modified to give a Ben-Or/Tiwari sparse interpolation algorithm over finite fields;
the early termination introduced by  Kaltofen and Lee \cite{KaltofenL03} is used to estimate the
terms bound for the coefficients of the substituted GCD;
a new type of good points is introduced to recover the GCD from its modular images
by using only primes with small sizes.
Combination of these ingredients leads to the lower binary complexity 
and the practical efficiency of the algorithm.

\subsection{Related work and comparison}

The EZ-GCD~\cite{MosesY73} appears to have a computing bound which in most
cases is a polynomial function of $T$ and $n$. But in some cases, the complexity is increased, for example,
when the number of terms in the expanded series form
of $B(x_1,x_2-b_2,\dots,x_n-b_n)$ has larger order than that   in $B(x_1,x_2,\dots,x_n)$ for some $(b_2,\dots,b_n)$ or when the input polynomials are not monic with respect to any variable.
%
In Zippel's algorithm \cite{Zippel79}, $O(ndT)$ images of the GCDs are needed, while our algorithm only need $O(nT)$ images and has a better complexity. In Table \ref{tab-1}, we list the complexities for the GCD algorithms compared with Zippel algorithm. Here we assume the probability of failure $\varepsilon$ is fixed.
%
The complexity is analysed by the authors of this paper.


Zippel's algorithm was originally designed for GCDs
which are monic w.r.t. the main variable.
The complexity of Zippel's algorithm  is also sensitive to the sparse representation.
%
The main advantage of our algorithm is that its complexity is linear in $D$,
while the complexity of Zippel's algorithm is quadratic in $d$.

\begin{table}[ht]
\footnotesize
\centering
\begin{tabular}{ccc}
 \hline
Algorithms &Bit complexity & Condition\\ \cline{1-3}
Zippel \cite{Zippel79}& $nd^2T\log q+ndT(T_A+T_B)\log d\log q$ & Monic GCD\\
This paper& $nDT(T_A+T_B)\log^2 q$  & All cases\\
 \hline
\end{tabular}
\caption{A ``soft-Oh'' complexity comparison  for GCD algorithms over $\F_q[x_1,\dots,x_n]$.
$T$ is the number of terms of the GCD.
Condition means under what condition the result is valid.}
\label{tab-1}
\end{table}

In the EEZ-GCD algorithm \cite{Wang80b}, a factorization of the GCD of the leading coefficient is computed. However, this step may lead to high complexity, because the number of terms of the factors may be very large.
Our algorithm predetermines the leading coefficient of the GCD by isolating the maximum term instead of factorization, which was similar to the method introduced by Cuyt and Lee \cite{cuyt2008new} and had controllable complexity.
Furthermore, inspired by the work of Klivans and Spielman \cite{klivans2001randomness},
we introduce a new variable substitution such that the leading coefficient of the substituted GCD is a monomial, which greatly decreases the computation cost.
%
%
Compared to the algorithm in \cite{hu2016fast},
 our algorithm also uses Ben-Or/Tiwari algorithm to interpolate the coefficients of GCD,
 but our algorithm was based on a new diversification method.
 Also, we do not use the Kronecker substitution and use only primes with small size.
These techniques allow us to give an exact bit complexity, while in their algorithm the smooth prime has  size $O(D^n)$ in theory.
Compared to the algorithm in \cite{tang2018computing},
our algorithm uses the method of isolating maximum term instead of the shifted homogenization introduced in \cite{cuyt2008new}
and our algorithm works for any finite field even if the degree is large.


\section{Basic concepts and preliminary results}

\subsection{Notations}
Let $\X=\{x_1,x_2,\ldots,x_n\}$ and
$$A(\X)=c_1M_1+c_2M_2+\cdots+c_tM_t\in\mathcal{F}[\X],$$ where $\mathcal{F} $ is any field, $c_i\neq 0$ and $M_i$'s are monomials. Let the exponent vector of $M_i$ be $\mathbf{e}_i=(e_{i,1},\dots,e_{i,n})$ and the monomials $M_i$'s are arranged in lexicographically increasing order of $\mathbf{e}_i$'s. Then $c_tM_t$ is called the {\em leading term} and $c_t$ is called the {\em leading coefficient},  denoted as $\LC(A)$.
%

We first introduce the concept of monomial content.
\begin{definition}\label{def-4}
Let $f\in \mathcal{F}[\X]$, where $\mathcal{F}$ is any field. Assume $f=\sum_{i=1}^tc_iM_i$, $c_i\neq 0$ and $M_i$ are distinct monomials. Then $\gcd(M_1,\dots,M_t)$ is called the {\em monomial content} of $f$, denoted  by $\MoCont(f)$. We call $f/\MoCont(f)$ the {\em monomial primitive part} of $f$ and denote it by $\MoPrim(f)$.
\end{definition}

Since $\MoCont(f)$ is the greatest common factor of $M_1,\dots,M_t$, $\MoCont(f)$ is a monomial and $\MoPrim(f)=f/\MoCont(f)$ is coprime to any monomial. In particular, $\MoPrim(f)$ is relatively prime to $\MoCont(f)$.

Let $A,B$ be  nonzero elements of $\mathcal{F}[\X]$.
Then $G=\gcd(A,B)$ is uniquely determined by assuming   $\LC(G)=1$.
%
%
We say that $A$ is similar to $B$, denoted as $A\approx B$, if there exists
an $a\in \mathcal{F}^*$ such that $a A=B$.

Let $\mathbf{s}=(s_1,\dots,s_n)\in\N_+^n$ be an integer vector.
Define
\begin{equation}\label{eq-1}
A_{(\mathbf{s},y)}:=\frac{A(x_1y^{s_1},\dots,x_ny^{s_n})}{y^k}
\end{equation}
where $k$ is the smallest exponent of $y$ in $A(x_1y^{s_1},\dots,x_ny^{s_n})$.
%
%
%
$A_{(\mathbf{s},y)}$ separates the terms of $A$ by degrees of $y$.

Let $A=A_1+\cdots+A_{\ell},$
where $A_i$ is a part of $A$ such that $\deg_yA_i(x_1y^{s_1},\dots,x_ny^{s_n})=d_i$. Assume $d_1<d_2<\cdots<d_{\ell}$.
Then $A_{(\mathbf{s},y)}=A_1+A_2y^{d_2-d_1}+\cdots+A_{\ell}y^{d_{\ell}-d_1}.$
Here $A_{\ell}$ is the leading coefficient of $A_{(\mathbf{s},y)}$ w.r.t $y$ if we regard $\mathcal{F}[\X]$ as the coefficients domain.
%

For any $A\in\F_q[\X]$ and $\overrightarrow{\alpha}=(\alpha_1,\dots,\alpha_n)\in \F^n_q$,
 denote  $A(\overrightarrow{\alpha})=A(\alpha_1,\dots,\alpha_n)$. For any $i\in\N$, denote $\overrightarrow{\alpha}^i=(\alpha^i_1,\dots,\alpha^i_n)$.
The main idea of the modular GCD algorithm is to interpolate $G=\gcd(A,B)$ from a sequence of evaluations.
Pick a sequence of evaluation points $\overrightarrow{\alpha}_1,\overrightarrow{\alpha}_2,\dots$ from $\F^n_q$, compute the images of $G$, then interpolate each part $G_i(\X)$ of $G$ from the scaled images.

$\F_q$ may not have enough elements, and in this case we work in a suitable extension $\F_q\subset \F_{q^m}$, where the latter one is represented as $\F_q[z]/\langle \Phi(z)\rangle$, for a degree-$m$ irreducible polynomial $\Phi$ over $\F_q$. With this representation, arithmetic operations in $\F_{q^m}$ can be done in $O^\sim(m)$ arithmetic operations in $\F_q$, and thus in $O^\sim(m \log q)$ bit operations.

The cost of sparse polynomial interpolations is determined mainly by the number of points $\overrightarrow{\alpha}_1,\overrightarrow{\alpha}_2,\dots,$ needed and the size of the prime power $q$ needed.

\subsection{Preliminary results}

We show that the monomial content and the monomial primitive part of the  GCD can be computed separately.
\begin{lemma}\label{lm-7}
Let $A,B,G\in \F_q[\X]$ and assume $G=\gcd(A,B)$. Then
\begin{itemize}
\item[(\lowercase\expandafter{\romannumeral1})]   $\MoCont(G)=\gcd(\MoCont(A),\MoCont(B))$, and
\item[(\lowercase\expandafter{\romannumeral2})]
$\MoPrim(G)=\gcd(\MoPrim(A),\MoPrim(B))$.
\end{itemize}
\end{lemma}
\begin{proof}
By Definition \ref{def-4}, we have $G=\MoCont(G)\cdot\MoPrim(G)$, where $\MoCont(G)$ is a monomial and \MoPrim(G) is a polynomial without any non-trivial monomial factors.
Then, $G=\gcd(A,B)=\gcd(\MoCont(A)\cdot\MoPrim(A),\MoCont(B)\cdot\MoPrim(B))=\gcd(\MoCont(A),\MoCont(B))\cdot
\gcd(\MoPrim(A),\MoPrim(B))$.
Here $\gcd(\MoCont(A),\MoCont(B))$ is a monomial. Due to the monomial primitivity of $\MoPrim(A)$ and $\MoPrim(B)$, $\gcd(\MoPrim(A),\MoPrim(B))$ is coprime to any monomial factors.
Due to the unique factorization of polynomials,  the lemma is proved.
\end{proof}

%
We should ensure that the GCD remains the same when the field is extended. Denote $G=\gcd_{\mathcal{F}}(A,B)$ as the GCD of $A,B$ over domain $\mathcal{F}[\X]$. The following result is well known.
\begin{lemma}
Assume $\mathcal{F}$ is a field, $A,B\in \mathcal{F}[\X]$. Let $G=\gcd_{\mathcal{F}}(A,B)$. For any extension field $\mathcal{K}\supset {\mathcal{F}}$, treat $A,B$ as the elements of $\mathcal{K}[\X]$. Then $G=\gcd_{\mathcal{K}}(A,B)$.
\end{lemma}
%

%

Our proof will make extensive use of the  Schwartz-Zippel Lemma.
\begin{lemma}\label{lm-6}\cite{Zippel79} Let $\mathcal{F}$ be a field and $A \in \mathcal{F}[\X]$
be non-zero with total degree $D$ and let $S\subset \mathcal{F}$ be a finite set. If $\overrightarrow{\beta}$ is chosen at random from $S^n$ then ${\rm Prob}[A(\overrightarrow{\beta})=0]\leq \frac{D}{|S|}$.
\end{lemma}

\subsubsection{Resultant}
Let $F_1=\sum_{i=0}^{d}a_iy^i$ and $F_2=\sum_{i=0}^{\ell} b_iy^i$. The $Sylvester\ matrix$ of $A,B$ is the $d+\ell$ by $d+\ell$ matrix

\begin{equation}
\left(
  \begin{array}{cccccccc}
    a_d & a_{d-1} & \cdots & a_1 & a_0 &  & \\
        & a_d & a_{d-1} & \cdots & a_1 & a_0 & \\
        &    & \cdots & \cdots & \cdots & \cdots &  \\
     &  &  & a_d & \cdots & \cdots & a_0  \\
    b_{\ell} & b_{\ell-1} & \cdots & b_1 & b_0 &  & \\
        & b_{\ell} & b_{\ell-1} & \cdots & b_1 & b_0 & \\
        &    & \cdots & \cdots & \cdots & \cdots &  \\
     &  &  & b_{\ell} & \cdots & \cdots & b_0  \\
   \end{array}
\right)
\end{equation}
where the upper part of the matrix consists of $\ell$ rows of coefficients of $F_1$, the lower part consists of $d$ rows of coefficients of $F_2$.
The resultant of $F_1$ and $F_2$ is the determinant of the Sylvester matrix of $F_1,F_2$, written as $\res_y(F_1,F_2)$.
The following are some facts. Denote $\LC_y(F_1)$ as the leading coefficient of $F_1$ w.r.t. $y$.

\begin{lemma}\cite{hu2016fast}
Let $\mathcal{D}$ be any integral domain and $F_1,F_2\in\mathcal{D}[y,\X]$. Let $R=\res_{y}(F_1,F_2),\overrightarrow{\alpha}\in \mathcal{D}^n$. Then

\begin{itemize}
\item[(\lowercase\expandafter{\romannumeral1})]   $R,\LC_y(F_1),\LC_y(F_2)$ are polynomials in $\mathcal{D}[\X]$, and
\item[(\lowercase\expandafter{\romannumeral2})]
If $\mathcal{D}$ is a field and $\LC_y(F_1)(\overrightarrow{\alpha})\cdot \LC_y(F_2)(\overrightarrow{\alpha})\neq 0$,
 then
    $\res_{y}(F_1(y,\overrightarrow{\alpha}),F_2(y,$ $\overrightarrow{\alpha}))=R(\overrightarrow{\alpha})$ and
    $\deg_{y}\gcd(F_1(y,\overrightarrow{\alpha}),F_2(y,\overrightarrow{\alpha}))>0\Longleftrightarrow \res_{y}(F_1(y,\overrightarrow{\alpha}),$ $F_2(y,\overrightarrow{\alpha}))=0$.
\end{itemize}
\end{lemma}

\begin{lemma}\label{the-14}
Let $A,B\in \mathcal{D}[y,\X],\mathbf{s}=(s_1,\dots,s_n)\in\N^n$ and $R=\res_y(A_{(\mathbf{s},y)},B_{(\mathbf{s},y)})$. Then $\deg R\leq 2\|\mathbf{s}\|_{\infty}\deg A\deg B$.
\end{lemma}

\begin{proof}
Assume $A_{(\mathbf{s},y)}=a_dy^d+\cdots+a_1y+a_0$ and $B_{(\mathbf{s},y)}=b_{\ell}y^{\ell}+\cdots+b_1y+b_0$, where $a_i,b_j\in \mathcal{D}[\X]$. By the definition of $A_{(\mathbf{s},y)}$ and $B_{(\mathbf{s},y)}$, we have $d\leq \|\mathbf{s}\|_{\infty}\deg A$ and $\ell\leq \|\mathbf{s}\|_{\infty}\deg B$. As the Sylvester matrix is $d+\ell$ by $d+\ell$ matrix and $\deg a_i\leq \deg A,\deg b_j\leq \deg B$, the degree of $R$ is no more than $\ell\cdot \deg A+d\cdot \deg B$, which is $\leq 2\|\mathbf{s}\|_{\infty}\deg A\deg B$.
\end{proof}

\subsection{Isolating the leading coefficient}
In this section, we will show how to find an $\s$ such that
the leading coefficient of $\gcd(A_{(\s,y)},B_{(\s,y)})$ in $y$ is a monomial.

\subsubsection{Generalized homogenization technique}
Let $A,B\in\F_q[\X]$.
Instead of directly computing the GCD of $A$ and $B$, we compute the GCD of
the {\em generalized homogenizing polynomials} $A_{(\mathbf{s},y)}$ and $B_{(\mathbf{s},y)}$ (see (\ref{eq-1}))
by introducing a new variable $y$. Then $A_{(\mathbf{s},y)},B_{(\mathbf{s},y)}\in \F_q[\X,y]$. Denote $C=\gcd(A_{(\mathbf{s},y)},B_{(\mathbf{s},y)})$ and $G=\gcd(A,B)$. The following lemma shows that $C\approx G_{(\mathbf{s},y)}$, which means $C$ and $G_{(\mathbf{s},y)}$ are the same up to a non-zero constant.

Denote $y^{\mathbf{s}}\X=(x_1y^{s_1},\dots,x_ny^{s_n})$ and $y^{-\mathbf{s}}\X=(x_1/y^{s_1},\dots,x_n/y^{s_n})$.

\begin{lemma}\label{the-5}
Let $A,B\in \F_q[\X]$, $G=\gcd(A,B)$, and $\mathbf{s}=(s_1,\dots,s_n)\in \N^n$.
Then $\gcd(A_{(\mathbf{s},y)},B_{(\mathbf{s},y)})\approx G_{(\mathbf{s},y)}$.
\end{lemma}
\begin{proof}
First we claim that $\gcd(A(y^\mathbf{s}\X),B(y^\mathbf{s}\X))\approx y^mG(y^\mathbf{s}\X)$ for some integer $m\geq 0.$
Proof of the claim:
Assume $P=\gcd(A(y^\mathbf{s}\X),B(y^\mathbf{s}\X))$.
$G|A$ and $G|B$ imply that $G(y^\mathbf{s}\X)|A(y^\mathbf{s}\X)$ and $G(y^\mathbf{s}\X)|B(y^\mathbf{s}\X)$, and then we have $G(y^\mathbf{s}\X)|P$.

We prove the reverse direction. Since $P|A(y^\mathbf{s}\X)$,
there exists a $Q\in\F_q[y,\X]$ such that $A(y^\mathbf{s}\X)=P(y,\X)Q(y,\X)$. Replacing $x_iy^{s_i}$ by $x_i$, we have
$A(\X)=P(y,y^{-\mathbf{s}}\X)$ $Q(y,y^{-\mathbf{s}}\X)$. Then, there exists an integer $k$ such that $Q(y,y^{-\mathbf{s}}\X)=y^k(Q_{\ell}y^{d_{\ell}}+\cdots+Q_1y^{d_1}+Q_0)$, where $Q_i\in \F_q[\X]$ and $d_i>0$. So $y^kP(y,y^{-\mathbf{s}}\X)$ is a polynomial in $\F_q[y,\X]$ and $y^kP(y,y^{-\mathbf{s}}\X)|A(\X)$.
If $k>0$, then $P(y,y^{-\mathbf{s}}\X)|A(\X)$. So we can always assume $k\leq 0$.
For the same reason, there exists an integer $u\leq 0$ such that $y^uP(y,y^{-\mathbf{s}}\X)|B(\X)$.
Now let $m'=\min(-k,-u)$.
Without loss of generality, assume $m'=-k$. Then $y^{-m'}P(y,y^{-\mathbf{s}}\X)|A(\X)$ and $y^{-m'}P(y,y^{-\mathbf{s}}\X)|y^{k-u}B(\X)$.
So $y^{-m'}P(y,y^{-\mathbf{s}}\X)|\gcd(A(\X),y^{k-u}B(\X))=\gcd(A,B)$,
which implies $P(y,y^{-\mathbf{s}}\X)|y^{m'}G$. Replace $x_i/y^{s_i}$ by $x_i$, we have $P(y,\X)|y^{m'}G(y^{\mathbf{s}}\X)$. So there exists an integer $m\geq 0$ such that $P\approx y^{m}G(y^{\mathbf{s}}\X)$.
The claim is proved.

Since $C=\gcd(A_{(\mathbf{s},y)},B_{(\mathbf{s},y)})$ and $A_{(\mathbf{s},y)}=\frac{A(y^\mathbf{s}\X)}{y^{d_A}}$, $Cy^{d_A}|A(y^\mathbf{s}\X)$.
Here $d_A$ is the integer $k$ in (\ref{eq-1}).
For the same reason, $Cy^{d_B}|B(y^\mathbf{s}\X)$.
So $$Cy^{\min\{d_A,d_B\}}|\gcd(A(y^\mathbf{s}\X),B(y^\mathbf{s}\X)).$$
By the claim, $\gcd(A(y^\mathbf{s}\X),B(y^\mathbf{s}\X))\approx y^mG(y^\mathbf{s}\X)$ for some integer $m\geq 0$,
so $$Cy^{\min\{d_A,d_B\}}|y^mG(y^\mathbf{s}\X).$$
Then
$C|\frac{y^mG(y^\mathbf{s}\X)}{y^{\min\{d_A,d_B\}}}.$
Clearly, $d_G\leq d_A$ and $d_G\leq d_B$. So
$d_G\leq \min\{d_A,d_B\}$ and  $C|\frac{y^mG(y^\mathbf{s}\X)}{y^{d_G}}=y^mG_{(\mathbf{s},y)}$.
For the reverse direction, since $G|A$ and $G|B$, we have $G(y^\mathbf{s}\X)|A(y^\mathbf{s}\X)$ and $G(y^\mathbf{s}\X)|B(y^\mathbf{s}\X)$.
Since $G(y^\mathbf{s}\X)=G_{(\mathbf{s},y)}\cdot y^{d_G}$ and $A(y^\mathbf{s}\X)=A_{(\mathbf{s},y)}\cdot y^{d_A}$,
$G_{(\mathbf{s},y)}|A_{(\mathbf{s},y)}$. For the same reason, we have $G_{(\mathbf{s},y)}|B_{(\mathbf{s},y)}$.
So we have $$G_{(\mathbf{s},y)}|\gcd(A_{(\mathbf{s},y)},B_{(\mathbf{s},y)})=C.$$
So there exists an integer $m'$ such that $C\approx y^{m'}G_{(\mathbf{s},y)}$. Regard $C$ and $G_{(\mathbf{s},y)}$ as polynomials in $y$ with coefficients in $\F_q[\X]$, we know both $C$ and $G_{(\mathbf{s},y)}$ have non-zero constants, so $m'=0$. The lemma is proved.
\end{proof}

Once $C=\gcd(A_{(\mathbf{s},y)},B_{(\mathbf{s},y)})$ is computed, the polynomial $C(1,\X)$ is similar to $\gcd(A,B)$.

\subsubsection{Isolating the leading coefficient}
In previous work on GCD computation, $A_{(\mathbf{1},y)}$ instead of $A_{(\mathbf{s},y)}$
is used, where $\mathbf{1}$ is the vector all of whose entries are $1$.
Suppose   $G=5x^3_1x_2+7x^5_1x^8_2+4x^9_1x^4_2$ is the GCD to be computed.
Then $G_{(\mathbf{1},y)} = 5x^3_1x_2+(7x^5_1x^8_2+4x^9_1x^4_2)y^9$.
Regarding $y$ as the main variable, $G_{(\mathbf{1},y)}$ is not monic. In this case, the sparse modular GCD algorithm of Zippel cannot be applied directly as the leading coefficient in the univariate images of $G$ in $y$ cannot be known in advance.
In the computing of GCD, how to find such a leading coefficient of $A_{(\mathbf{1},y)}$
is a key and bottleneck step.
On the other hand, let $\mathbf{s}=(1,2)$. Then 
the leading coefficient of $G_{(\mathbf{s},y)}=5x^3_1x_2+4x^9_1x^4_2y^{12}+7x^5_1x^8_2y^{16}$ in $y$
is a monomial, which will be used to greatly simplify the GCD computation.

In this section, we will introduce a new method to solve this leading coefficient problem. We know that $\LC_y(G)$ divides $\gcd(\LC_y(A), \LC_y(B))$.
If a new variable $y$ is constructed so that $\gcd(\LC_y(A),\LC_y(B))$ is only a monomial, then $\LC_y(G)$ must also be a monomial. In order to make $\gcd(\LC_y(A),\LC_y(B))$ a monomial, the simplest case is that $\LC_y(A)$ or $\LC_y(B)$ is a monomial.

Before our description, we define the concept of the maximum isolated term.
\begin{definition}
Let $F=f_{\ell}y^{e_{\ell}}+f_{\ell-1}y^{e_{\ell-1}}+\cdots+f_1y^{e_1}\in \F_q[\X,y]$,
where $f_i\in \F_q[\X],f_i\neq 0$ and $e_{\ell}>\cdots>e_1\geq 0$. If $f_{\ell}$ is a single term in $\F_q[\X]$, then we say $F$ has a maximum isolated term w.r.t $y$.
\end{definition}

The following lemma says that if a polynomial has a maximum isolated term w.r.t $y$, then so do its factors.
\begin{lemma}
If $F\in \F_q[\X,y]$ has a maximum isolated term w.r.t $y$, then its factor polynomials
also have maximum isolated terms w.r.t $y$.
\end{lemma}
\begin{proof}
Assume $F=G\cdot H$ and $G=g_{\ell}y^{d_{\ell}}+\cdots+g_1y^{d_1}$ and $H=h_ty^{e_t}+\cdots+h_1t^{e_1}$. Then the leading coefficient of $F$ is $g_{\ell}\cdot h_t$. If the number of terms of $g_{\ell}$ or $h_t$ exceeds one, so does  $g_{\ell}\cdot h_t$, which contradicts to assumption of $F$.
\end{proof}

The following theorem gives a probabilistic method to construct  $\mathbf{s}$, so that the new polynomial has a maximum isolated term. 

\begin{theorem}\label{the-4}
Let $A(\X)\in\F_q[\X]$, $T\geq\#A$, $N=2(T-1)$.
If we choose a vector $\mathbf{s}=(s_1,\dots,s_n)\in [1,N]^n$ uniformly at random, then $A_{(\mathbf{s},y)}$ has a maximum isolated term w.r.t $y$  with probability $\geq \frac12$.
\end{theorem}
\begin{proof}
Assume $A=\sum_{i=1}^t a_ix_1^{e_{i,1}}\cdots x_n^{e_{i,n}}$.
The degrees of $y$ of terms in $A(y^\mathbf{s}\X)$
for $\mathbf{s}=(s_1,\dots,s_n)$ are $d_{\s,i}=e_{i,1}s_1+\cdots+e_{i,n}s_n,i=1,\dots,t$.
Let $d_{\s,\max}=\max_{i=1}^t  d_{\s,i}$ and call  $(s_1,\dots,s_n,d_{\s,\max})$ the maximum point of $\mathbf{s}$.

Considering $s_1,\dots,s_n,z$ as variables, we have $t$ hyperplanes $P_i:z=e_{i,1}s_1+\cdots e_{i,n}s_n,i=1,\ldots,t$ in $\R^{n+1}$.
%
%
Let $\overline{S}=\{\s\in\R^n\,:\, s_i>0,i=1,\ldots,n \}$ be the open first octant.
Define $\PC\subset\R_+^{t+1}$ as follows.
$$\PC=\{(\s,d_{\s,i_m})\,:\, \s\in\overline{S} \hbox{ and } i_m\in\argmax_{i=1}^t d_{\s,i} \}$$
that is, $\PC$ consists of maximum points over $\overline{S}$.

We claim that $\PC=\cup_{i=1}^{\ell} Q_i$ is an open $n$-dimensional polyhedral cone,
where $\ell\le t$, $Q_i\subset P_{\mu_i}$ is a convex polyhedral cone, $P_{\mu_i}\ne P_{\mu_j}$ for $i\ne j$,
$Q_i\cap Q_{i+1}\subset\PC$ for $i=1,\ldots,\ell-1$.
Furthermore, the map $\DC(\s)=d_{\s,i_m}: \overline{S}\rightarrow \R$ for $(\s,d_{\s,i_m})\in\PC$
is a concave function.

We prove the claim by induction.
The claim is easily seen to be true for $t=1$.
For $t=2$, let the projection of the intersection of $P_1 $ and $P_2$ to the
$\s$-coordinate space be $R_1=\{\s\,:\, P_1(\s) = P_2(\s)\}$ which is a linear
subspace of the  $\s$-space $\R^n$.
If $R_1$ is outside $\overline{S}$, then we have either
$P_1(\s)>P_2(\s)$ for all $\s\in\overline{S}$
or
$P_1(\s)<P_2(\s)$ for all $\s\in\overline{S}$.
We can set $\PC = \{(\s,P_1(\s))\,:\, \s\in\overline{S}\}$ in the first case
and  $\PC = \{(\s,P_2(\s))\,:\, \s\in\overline{S}\}$ in the second case,
and the claim is proved.
If $R_1$ is inside $\overline{S}$, then
$\overline{S}$ is divided into two convex polyhedral cones:
$C_1=\{\s\,:\, d_{\s,1} \ge d_{\s,2} \}$ and
$C_2=\{\s\,:\, d_{\s,2} \ge d_{\s,1}\}$ by $R_1$.
It is clear that $C_1\cap C_2 =R_1$.
Let $Q_i = \{(\s,d_{\s,i})\,:\, \s\in C_i \}$, which are clearly convex polyhedral cones.
Then it is easy to see that $\PC=Q_1\cup Q_2$.
Since all coordinates of $\s$ are positive,
$\DC(\s)$ is clearly concave.
Also note that $A_{(\s,y)}$ has a maximum isolated term for
$\s\in\overline{S}\setminus R_1$.

Suppose the claim is valid for $t$ and $A$ has $t+1$ monomials.
Then for the first $t$ monomials of $A$, $\PC_t=\cup_{i=1}^\ell Q_i$ with $\ell\le t$.
Let $i_1$ be the smallest index such that $P_{t+1}$ intersects $Q_{i_1}$
and $i_2\geq i_1+1$ be the next smallest index such that $P_{t+1}$ intersects $Q_{i_2}$.
Here, we consider the generic case, that is  $Q_{i_1}\cap Q_{i_1+1}\subset P_{t+1}$
and $Q_{i_2}\cap Q_{i_2+1}\subset P_{t+1}$  are not valid.
If one of them is valid, the claim can be proved similarly.
Also, $P_{t+1}$ may intersect only one $Q_{i}$, and this case can also be proved similarly.

Since $\DC_t(\s)$  is concave, $P_{t+1}$ intersects
no $Q_{i}$ for $i\geq i_2+1$,
that is $P_{t+1}$ intersects essentially at most two $Q_i$s.
Let $E_i (i=1,\ldots,t)$,  $R_1$, and $R_{2}$ be the projections of $Q_{i}$,
$P_{t+1}\cap  Q_{i_1}$ and $P_{t+1}\cap  Q_{i_2}$
to the $\s$-coordinate space.
Further let
\begin{equation*}
\label{eq-def11}
\begin{array}{ll}
C_1=\{\s\in E_{i_1}\,:\, d_{\s,\mu_{i_1}} \ge d_{\s,t+1} \} &
\widetilde{Q}_{i_1}=\{(\s,d_{\s,\mu_{i_1}}) \,:\, \s\in C_1 \} \\
C_2=\{\s\in E_{i_2}\,:\, d_{\s,\mu_{i_2}} \ge d_{\s,t+1}\}&
\widetilde{Q}_{i_2}=\{(\s,d_{\s,\mu_{i_2}}) \,:\, \s\in C_2 \} \\
C_3=\{\s\in \cup_{i=i_1}^{i_2} E_{i}\,:\, d_{\s,t+1} \ge d_{\s,\mu_i}\} &
\widetilde{Q}_{i_3}=\{(\s,d_{\s,t+1}) \,:\, \s\in C_3 \}.
\end{array}
\end{equation*}
Since $\DC_t(\s)$ is concave, we have $d_{\s,t+1} \ge d_{\s,\mu_i}$ for $i=i_1+1,\ldots,i_2-1$.
Then, it can be shown that the following decomposition satisfies the properties in the claim
$$\PC_{t+1} =Q_1\cup\cdots Q_{i_1-1}\cup \widetilde{Q}_{i_1}\cup \widetilde{Q}_{i_3}\cup \widetilde{Q}_{i_2}
\cup Q_{i_2+1}\cup\cdots Q_{\ell}.$$

Let $F_i (i=1,\ldots,\ell-1)$ be the projections of $Q_{i}\cap Q_{i+1}$ to the $\s$-coordinate space.
Then, for $\s \in \overline{S}\setminus \cup_{i=1}^{\ell-1} F_i$,
$A_{(\s,y)}$ has a maximum isolated term.
%
%
%
%
Define the polynomial $B(\s) = \prod_{i=1}^{\ell-1}(d_{\s,\mu_i}-d_{\s,\mu_{i+1}})$.
Then
$\deg B(\s) \le T-1$. By Lemma \ref{lm-6}, if randomly choose $\mathbf{s}\in [1,N]^n$, with the probability $\geq 1-\frac{\ell}{N}\geq 1-\frac{T-1}{2(T-1)}=\frac12$, $\mathbf{s}$ is not a zero of the $B(\s)$, and in this case, $A_{(\mathbf{s},y)}$ has a maximum isolated term.
\end{proof}

\begin{example}
\label{ex-11}
Let $A=2x_1^7x_2^3+3x_1^5x_2^8+5x_1x_2^9\in\F_{11}[x_1,x_2]$. For $\mathbf{s}\in \N^2$, we have the degrees of $y$ in $A(x_1y^{s_1},x_2y^{s_2})$ are  $d_1=7s_1+3s_2,d_2=5s_1+8s_2,d_3=s_1+9s_2$.
Regarding $s_1,s_2,z$ as variables,  we obtain three hyperplanes:
\[
\begin{cases}
P_1:z=7s_1+3s_2\\
P_2:z=5s_1+8s_2\\
P_3:z=s_1+9s_2
\end{cases}
\]
As shown Figure \ref{figrdegree},  $P_1,P_2,P_3$ form an open polyhedral cone $\PC$, which is concave as a function of $(s_1,s_2)$.
The  projection of the edges of $\PC$ to the $s_1s_2$-coordinate plane are two lines  $7s_1+3s_2=5s_1+8s_2$ and $5s_1+8s_2=s_1+9s_2$, shown in Figure \ref{figrdegree2}.
Over these two lines, two of $P_1,P_2,P_3$ achieve the same maximum value for a given $\mathbf{s}$.
Thus $A_{(\s,y)}$ has a maximum isolated term if and only if
$7s_1+3s_2\neq 5s_1+8s_2$ and $5s_1+8s_2\neq s_1+9s_2$.

\begin{figure}[!hptb]
\begin{minipage}[t]{0.49\linewidth}
\centering
\includegraphics[scale=0.22]{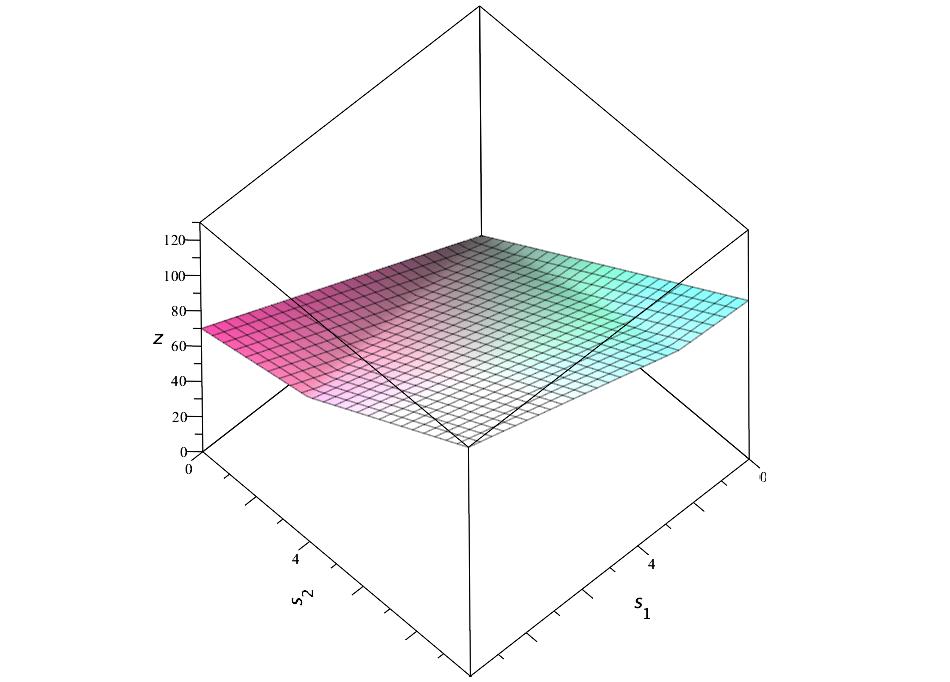}
\caption{The open polyhedral cone $\PC$ formed by $P_1,P_2,P_3$.}
\label{figrdegree}
\end{minipage}
\begin{minipage}[t]{0.49\linewidth}
\centering
\includegraphics[scale=0.50]{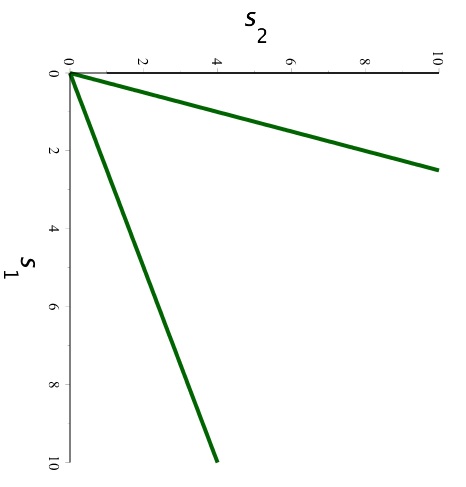}
\caption{Projection of edges of the open polyhedral cone}\label{figrdegree2}
\end{minipage}
\end{figure}
\end{example}

For polynomials $A$ and $B$,  we can always choose a vector $\mathbf{s}$ such that $A_{(\mathbf{s},y)}$ or $B_{(\mathbf{s},y)}$ has a maximum isolated term. As $N=2(T_A-1)$ or $N=2(T_B-1)$, the degrees of $A_{(\mathbf{s},y)}$ or $B_{(\mathbf{s},y)}$  in $y$ are $O(T_AD)$ or $O(T_BD)$. Once one of $A$ and $B$
has a maximum isolated term, so does their GCD.

\subsection{Diverse polynomial and sparse interpolation over finite fields}
In this section, we give the Ben-Or/Tiwari sparse interpolation over finite fields.

\subsubsection{Diverse polynomials}
\label{sec-dp}
We   use the following concept of diverse polynomials, introduced by Giesbrecht and Roche \cite{giesbrechtD11}.
%
\begin{definition}\label{def-1}
Let $\mathcal{R}$ be any ring. If a polynomial $f\in \mathcal{R}[\X]$ has all coefficients distinct; that is, $f=\sum_{i=1}^tc_iM_i$ and $c_i=c_j\Rightarrow i=j$, then we say $f$ is {\em diverse}.
\end{definition}

We define the following more loosely concept: diverse w.r.t. $y$.
\begin{definition}
Let $F\in \mathcal{R}[y,\X]$. Assume
$F=\sum_{i=0}^{d}a_iy^i.$
$F$ is called {\em diverse} w.r.t. $y$ if each $a_i\in \mathcal{R}[\X]$ is diverse.
\end{definition}

The following is an illustrative  example for this concept.
\begin{example}
Let $F=(4x_1x^2_2+6x^2_1x^5_2)y^3+(3x^2_1x_2+2x^2_1x^2_2)y\in \F_7[y,x_1,x_2].$
Regard $y$ as the main variable in $F$. The coefficients of $y^3$ and $y$ are $4x_1x^2_2+6x^2_1x^5_2$ and $3x^2_1x_2+2x^2_1x^2_2$. Both of them have the pair-wise different coefficients, so $F$ is $diverse$ w.r.t. $y$.
As a counter-example, if $F=(6x_1x^2_2+6x^2_1x^5_2)y^3+(3x^2_1x_2+2x^2_1x^2_2)y\in\F_7[y,x_1,x_2],$
$6x_1x^2_2+6x^2_1x^5_2$, the coefficient of $y^3$, has the same coefficient $6$. So $F$ is not diverse w.r.t $y$.
\end{example}

Giesbrecht and Roche \cite{giesbrechtD11} introduced a method of
{\em diversification}, which converted a non-diverse polynomial into a diverse polynomial with high probability.
If $\zeta_1,\dots,\zeta_n\in \mathcal{R}^*$, the polynomials
$f(\zeta_1x_1,\dots,\zeta_nx_n)\leftrightarrow f(\X)$
 are one-to-one corresponding. We can interpolate $f(\zeta_1x_1,\dots,\zeta_nx_n)$ instead of $f(\X)$.
If $f(\X)=\sum_{i=1}^tc_ix^{e_{i,1}}_1\cdots x^{e_{i,n}}_n,$
then $$f(\zeta_1x_1,\dots,\zeta_nx_n)=\sum_{i=1}^tc_i\zeta^{e_{i,1}}_1\cdots \zeta^{e_{i,n}}_nx^{e_{i,1}}_1\cdots x^{e_{i,n}}_n=\sum_{i=1}^t\widetilde{c}_ix^{e_{i,1}}_1\cdots x^{e_{i,n}}_n,$$
where
$\widetilde{c}_i=c_i\zeta^{e_{i,1}}_1\cdots \zeta^{e_{i,n}}_n.$
Now the coefficients of $f(\zeta_1x_1,\dots,\zeta_nx_n)$ are $\widetilde{c}_i$'s.
Giesbrecht and Roche \cite{giesbrechtD11} proved that if $\mathcal{R}$ has enough many elements and $(\zeta_1,\dots,\zeta_n)$ are randomly chosen from $\mathcal{R}^{*n}$, then $f(\zeta_1x_1,\dots,\zeta_nx_n)$ is diverse with high probability.
Once $g=f(\zeta_1x_1,\dots,\zeta_nx_n)$ is known, so $f=g(\zeta^{-1}_1x_1,\dots,\zeta^{-1}_nx_n).$

The following theorem states that diversification of   $F_1$ and $F_2$  leads to diversification of their GCD.
Denote $\overrightarrow{\zeta}\X=(\zeta_1x_1,\dots,\zeta_nx_n)$ and $\overrightarrow{\zeta}^{-1}=(\zeta^{-1}_1,\dots,\zeta^{-1}_n)$.

\begin{lemma}\label{the-3}
Let $F_1,F_2\in \F_q[y,\X]$, $C=\gcd(F_1,F_2)$, and $\overrightarrow{\zeta}=(\zeta_1,\dots,\zeta_n)\in \mathcal{K}^{*n}$, where $\mathcal{K}$ is an extension field of $\F_q$. Then $C(y,\overrightarrow{\zeta} \X )\approx\gcd(F_1(y,\overrightarrow{\zeta} \X),F_2(y,\overrightarrow{\zeta} \X))$.
\end{lemma}
\begin{proof}
Assume $P=\gcd(F_1(y,\overrightarrow{\zeta} \X),F_2(y,\overrightarrow{\zeta} \X))$.
Since $C=\gcd(F_1,F_2)$, we have $C(y,\overrightarrow{\zeta} \X)|F_1(y,\overrightarrow{\zeta} \X)$ and $C(y,\overrightarrow{\zeta} \X)|F_2(y,\overrightarrow{\zeta} \X)$. So we have $C(y,\overrightarrow{\zeta} \X)|P$.
We prove the reverse direction.
From $P|F_1(y,\overrightarrow{\zeta} \X)$, we have
$P(y,\overrightarrow{\zeta}^{-1} \X)|F_1$. For the similar reason, $P(y,\overrightarrow{\zeta}^{-1} \X)|F_2$, which implies $P(y,\overrightarrow{\zeta}^{-1} \X)|\gcd(F_1,F_2)$. So $P(y,\overrightarrow{\zeta}^{-1} \X)|C$ and
then $P|C(y,\overrightarrow{\zeta}\X)$. The lemma is proved.
\end{proof}

\subsubsection{Sparse interpolation over finite fields}
We generalize the Ben-Or and Tiwari algorithm to polynomials over finite fields.
Compared with the original Ben-Or and Tiwari algorithm over fields with characteristic $0$, the
 following assumption need to be satisfied.
\begin{assumption}\label{ass-2}
Let $f=\sum_i c_iM_i$, $M_i=x_1^{e_{i,1}}\cdots x_n^{e_{i,n}}$,
$\overrightarrow{\alpha} =(\alpha_1,\dots,\alpha_n) \in\F_{q^m}^n$, and $\omega$ a primitive root of $\F_q$.
\begin{enumerate}
\item $f$ is a diverse polynomial.
\item The polynomial $\prod_{1\leq i<j\leq t}(M_i-M_j)$ is not zero at point $\overrightarrow{\alpha}$.
\item The polynomial $\prod_{k=1}^n\prod_{1\leq i<j\leq t}(\omega^{e_{i,k}}M_i-\omega^{e_{j,k}}M_j)$ is not zero at point $\overrightarrow{\alpha}$.
\end{enumerate}
\end{assumption}

The algorithm is listed below for ease of the call of other algorithms.
For details, see  \cite{huang21}.
 \begin{algorithm}
 Interpolation
\label{alg-2}

{\noindent\bf Input:}
\begin{itemize}
\item $2T$ evaluations $f(\overrightarrow{\alpha}^i)=f(\alpha^i_1,\dots,\alpha^i_n),i=1,2,\dots,2T$, where Assumptions \ref{ass-2} are satisfied.
\item A primitive root $\omega$ of $\F_q$, where $q>\max_{i=1}^n\deg_{x_i}f$.
\item $2nT$ evaluations $f(\overrightarrow{\alpha}^i_k)=f(\alpha^i_1,\dots,\alpha^{i}_{k-1},(\alpha_k\omega)^i,\alpha^{i}_{k+1},\dots,\alpha^i_n),i=1,2,$ $\dots,2T,k=1,2,\dots,n$.
\end{itemize}

{\noindent\bf Output:}
The polynomial $f=\sum_{i=1}^tc_iM_i$.
\end{algorithm}

We choose the points in the extension field $\F_{q^m}=\F_q[z]/(\Phi)$ with irreducible polynomial $\Phi(z)$ of degree $m$.

\begin{theorem}\label{the-2}\cite{huang21}
Algorithm \ref{alg-2} needs $O^\sim(nm^2T\log^2 q+nT\sqrt{d}\log q)$ bit operations.
\end{theorem}

\begin{remark}
The complexity $nT\sqrt{d}\log q$ comes from the computing of discrete logarithms.
\end{remark}

\subsection{Early termination for the terms bound}
\label{sec-et}
We show how to estimate a tight terms bound for a polynomial.
Kaltofen and Lee \cite{KaltofenL03} proposed the technique {\em early termination}, which can be used to detect the number of terms of $f$ with high probability. Based on this idea, a method that to test whether $f$ is $t$-sparse is given.
Hu and Monagan \cite{hu2021fast} also applied this method to determine the terms bound of the GCD. Let $v_i =f(x_1^i,\dots,x_n^i )$ be the symbolic evaluations of $f$ at powers, and define the Hankel matrices of polynomials

\begin{equation*}
\HK_s=\left(
\begin{array}{cccc}
v_1 & v_2 & \cdots & v_s  \\
v_2 & v_3 & \cdots &v_{s+1} \\
\vdots & \vdots & \ddots & \vdots\\
v_s & v_{s+1} &\cdots & v_{2s-1}\\
\end{array}
\right)
\end{equation*}

Kaltofen and Lee \cite{KaltofenL03} proved that if $s>t$, then $\HK_s$ is singular; if $s\leq t$, $\HK_s$ has full rank.
For any $\overrightarrow{\alpha}$ with components taken from the algebraic completion of $\F_q$, we have
\[\det\HK_s(\overrightarrow{\alpha})\begin{cases}
=0,&\text{if $s>t$},\\
\neq 0\ \text{with high probability,}&
\text{if $s\leq t$}.
\end{cases}\]

The degree of $\det\HK_s$ is bounded by $s^2\deg f$ \cite[Theorem 5]{KaltofenL03}.
If $s\leq t$, for $\overrightarrow{\alpha}$ chosen uniformly at
random from $\F_{q^m}^n$, $\det\HK_s(\overrightarrow{\alpha})$ is nonzero with probability at least $1-s^2\deg f/q^m$ by Lemma \ref{lm-6}.
Choose a random point $\overrightarrow{\alpha} \in \F_{q^m}^n$, if we test whether $\det\HK_s(\overrightarrow{\alpha})=0$ for $s=1,2,3,\dots,t+1,$ the probability that
$\det\HK_{t+1}(\overrightarrow{\alpha})$ is not the first singular Hankel matrix is at most
$\frac{\deg f}{q^m} \sum^t_{s=1}s^2=\frac{t(t+1)(2t+1)\deg f}{6q^m}$.

In this paper, an estimation for $t$ that is tight up to a constant factor is enough. As shown by Arnold \cite{arnold2016sparse}, in this case one can employ a technique called repeated doubling. We make an initial guess $s=1$, and test $\det\HK_s(\overrightarrow{\alpha})=0$ for
$s=1, 2, 2^2,\dots $, until $\det\HK_s(\overrightarrow{\alpha})=0$. In this case, the probability that the first instance of $\det\HK_s(\overrightarrow{\alpha})=0$ is for $t<s\leq 2t$ is
$\frac{\deg f}{q^m}\sum_{i=0}^{\lfloor \log_2 t\rfloor} (2^i)^2 <\frac{4t^2\deg f}{3q^m}.$

\subsection{Good point}
\label{sec-gp}
The main idea of our algorithm is mapping the entire problem to a simpler
domain via homomorphisms. Assume
$$\Phi_{\overrightarrow{\alpha}}:\F_q[\X,y]\rightarrow \F_q[y]$$
is a homomorphism of rings by evaluating $x_i=\alpha_i,i=1,\dots,n$, where $\overrightarrow{\alpha}=(\alpha_1,\dots,\alpha_n)$.
Let $F_1,F_2\in \F_q[\X,y]$ and $C=\gcd(F_1,F_2)$. Then $\Phi_{\overrightarrow{\alpha}}(F_1)=F_1(y,\overrightarrow{\alpha})$ and $\Phi_{\overrightarrow{\alpha}}(F_2)=F_2(y,\overrightarrow{\alpha})$. Compute the GCD of univariate polynomials $\Phi_{\overrightarrow{\alpha}}(F_1)$ and $\Phi_{\overrightarrow{\alpha}}(F_2)$, and it is
easy to see that
$$\Phi_{\overrightarrow{\alpha}}(C)|\gcd(\Phi_{\overrightarrow{\alpha}}(F_1),\Phi_{\overrightarrow{\alpha}}(F_2)).$$

If $\Phi_{\overrightarrow{\alpha}}(C)\approx\gcd(\Phi_{\overrightarrow{\alpha}}(F_1),\Phi_{\overrightarrow{\alpha}}(F_2))$,
$\gcd(\Phi_{\overrightarrow{\alpha}}(F_1),\Phi_{\overrightarrow{\alpha}}(F_2))$ retains parts of the information of $C$ to solve the problem in the original domain.
The leading coefficient $\LC_y(C)(\overrightarrow{\alpha})$ of the GCD is not zero
   if $\Phi_{\overrightarrow{\alpha}}(F_1)$ and $\Phi_{\overrightarrow{\alpha}}(F_2)$ do not
decrease in degree, which leads to the following definition.

\begin{definition}\label{def-2}
Let $F_1,F_2\in \F_q[\X,y]$ and $C=\gcd(F_1,F_2)$. Let $\overrightarrow{\alpha}\in \F^n_{q^m}$. We say $\overrightarrow{\alpha}$ is a {\em good point} for $F_1,F_2$ if $\LC_y(F_1)(\overrightarrow{\alpha}) \neq 0$, $\LC_y(F_2)(\overrightarrow{\alpha}) \neq 0$ and $\deg \Phi_{\overrightarrow{\alpha}}(C) =\deg \gcd(\Phi_{\overrightarrow{\alpha}}(F_1),\Phi_{\overrightarrow{\alpha}}(F_2))$.
\end{definition}

Since $\Phi_{\overrightarrow{\alpha}}(C)|\gcd(\Phi_{\overrightarrow{\alpha}}(F_1),\Phi_{\overrightarrow{\alpha}}(F_2))$,
if $\overrightarrow{\alpha}$ is a good point for $F_1,F_2$, we always have $\Phi_{\overrightarrow{\alpha}}(C)\approx\gcd(\Phi_{\overrightarrow{\alpha}}(F_1),\Phi_{\overrightarrow{\alpha}}(F_2))$.

\begin{lemma}\label{the-13}
Let $F_1,F_2\in \F_q[\X,y]$ and $\overrightarrow{\alpha}\in \F^n_{q^m}$. Assume $$R=\res_y(F_1/\gcd(F_1,F_2),F_2/\gcd(F_1,F_2)).$$ Let $L=R\cdot \LC_y(F_1)\cdot \LC_y(F_2)$. Then
$\overrightarrow{\alpha}$ is a good point for $F_1,F_2$ if and only if $L(\overrightarrow{\alpha})\neq 0$.
\end{lemma}
\begin{proof}
First, we have $R=\res_y(F_1/\gcd(F_1,F_2),F_2/\gcd(F_1,F_2))\in \F_q[\X]$. Since $\gcd(F_1/\gcd(F_1,F_2),F_2/\gcd(F_1,F_2))=1$, $R\neq 0$.

Assume $L(\overrightarrow{\alpha})\neq 0$. Then $\LC_y(F_1)(\overrightarrow{\alpha})\neq 0, \LC_y(F_2)(\overrightarrow{\alpha})\neq 0$ and $R(\overrightarrow{\alpha})\neq 0$. Assume $C=\gcd(F_1,F_2)$. As the leading coefficients of $F_1,F_2$ are not zero at point $\overrightarrow{\alpha}$, by the definition of resultant, $R(\alpha)=\res_y(F_1(y,\overrightarrow{\alpha})/C(y,\overrightarrow{\alpha}),F_2(y,\overrightarrow{\alpha})/C(y,\overrightarrow{\alpha}))$.
Since $R(\alpha)\neq 0$, $\gcd(F_1(y,\overrightarrow{\alpha})/C(y,\overrightarrow{\alpha}),F_2(y,\overrightarrow{\alpha})/C(y,\overrightarrow{\alpha}))=1$.
So $C(y,\overrightarrow{\alpha})\approx \gcd(F_1(y,\overrightarrow{\alpha}),F_2(y,\overrightarrow{\alpha}))$, which implies $\deg \Phi_{\overrightarrow{\alpha}}(C) =\deg \gcd(\Phi_{\overrightarrow{\alpha}}(F_1),\Phi_{\overrightarrow{\alpha}}(F_2))$.
So $\overrightarrow{\alpha}$ is a good point for $F_1,F_2$. For the other direction, assume $\overrightarrow{\alpha}$ is a good point for $F_1,F_2$, then $\LC(F_1)(\overrightarrow{\alpha})\neq 0, \LC(F_2)(\overrightarrow{\alpha})\neq 0$. The remaining proof can be traced back directly.
\end{proof}

Of course, a single $\Phi_{\overrightarrow{\alpha}}$ does not usually retain all the information
necessary to solve the problem in the original domain.
For a fixed primitive root $\omega$ of $\F_q$, we have the following definition.
\begin{definition}\label{def-3}
Let $\overrightarrow{\alpha}\in \F^n_{q^m}$. If the points $\overrightarrow{\alpha}^i,i=1,2,\dots,2T$ and $\overrightarrow{\alpha}_k^i,i=1,2,\dots,2T,k=1,2,\dots,n$ (the element at the $k$-th row and $i$-th column of Table \ref{tab-2}) are all {\em good points} for $A$ and $B$, then
$\overrightarrow{\alpha}$ is called a $2T$-$\omega$ good point for $A$ and $B$.
\end{definition}

\noindent
\begin{table}[ht]
\footnotesize
\centering
\begin{tabular}{c|c|c|c|c}
Base&$(\alpha_1,\alpha_2,\dots,\alpha_n)$&$(\alpha^2_1,\alpha^2_2,\dots,\alpha^2_n)$&$\dots$ &$(\alpha^{2T}_1,\alpha_2^{2T},\dots,\alpha_n^{2T})$\\ \cline{1-5}

1& $(\alpha_1\omega,\alpha_2,\dots,\alpha_n)$ & $((\alpha_1\omega)^2,\alpha^2_2,\dots,\alpha^2_n)$ & $\dots$  &$((\alpha_1\omega)^{2T},\alpha_2^{2T},\dots,\alpha_n^{2T})$\\

$2$ & $(\alpha_1,\alpha_2\omega,\dots,\alpha_n)$ & $(\alpha^2_1,(\alpha_2\omega)^2,\dots,\alpha^2_n)$ & $\dots$  & $(\alpha^{2T}_1,(\alpha_2\omega)^{2T},\dots,\alpha_n^{2T})$ \\

$\vdots$ & $\vdots$ & $\vdots$ & $\vdots$ & $\vdots$\\

$n$&$(\alpha_1,\alpha_2,\dots,\alpha_n\omega)$ & $(\alpha^2_1,\alpha^2_2,\dots,(\alpha_n\omega)^2)$ & $\dots$ & $(\alpha^{2T}_1,\alpha_2^{2T},\dots,(\alpha_n\omega)^{2T})$\\
\end{tabular}
\caption{The evaluation points for the $2T$-$\omega$ good point}\label{tab-2}
\end{table}

%
Our GCD algorithm cannot reconstruct $\gcd(F_1,F_2)$ using the images if $\overrightarrow{\alpha}$ is not a $2T$-$\omega$ good point for $F_1$ and $F_2$.

\section{A GCD algorithm over finite field}

In this section, we present a GCD algorithm for polynomials over finite fields.
Lemma \ref{lm-7} shows that
\begin{equation}\label{eq-2}
\small
G=\gcd(A,B)=\gcd(\MoCont(A),\MoCont(B))\cdot \gcd(\MoPrim(A),\MoPrim(B)),
\end{equation}
where $\gcd(\MoCont(A),\MoCont(B))$ and
$\gcd(\MoPrim(A),\MoPrim(B))$ are the monomial content of $G$ and the monomial primitive part of $G$, respectively.
Based on Equ. (\ref{eq-2}), to compute $G$, our algorithm is mainly divided into three parts.

\begin{description}
\item[Part 1] compute the monomial content of $G$ by computing $$\MoCont(G)=\gcd(\MoCont(A),\MoCont(B)).$$ This part is trivial as all polynomials appearing in the computing are monomials.

\item[Part 2] compute the monomial primitive part of $G$ by computing $$\MoPrim(G)=\gcd(\MoPrim(A),\MoPrim(B)).$$ We explain the framework of this part more details in Section \ref{sec-a3} and summarize it as a subroutine algorithm (given in Section \ref{sec-a2}).

\item[Part 3] multiply the two parts to get the final result $G=\MoCont(G)\cdot \MoPrim(G)$. This part can be converted to the additions of exponents, as $\MoCont(G)$ is a monomial.
\end{description}

\subsection{Framework of our GCD algorithm for monomial primitive polynomials}\label{sec-a3}
Assume $A,B\in \F_q[\X]$ are monomial primitive.
To compute $G=\gcd(A,B)$, we first compute $C=\gcd(A_{(\mathbf{s},y)},B_{(\mathbf{s},y)})$ for some suitable vector $\mathbf{s}\in\N^n$, and then let $y=1$ to obtain $G$ from $C$.
The algorithm is mainly divided into following parts:

\begin{description}
\item[Part a] find a vector $\mathbf{s}\in\N^n$ such that $G_{(\mathbf{s},y)}$ has a maximum isolated term w.r.t $y$ using Theorem \ref{the-4}.

\item[Part b] compute the polynomial $H=\Delta\cdot\gcd(A_{(\mathbf{s},y)},B_{(\mathbf{s},y)})$ by evaluation-interpolation scheme. Here $H$ is the GCD of $A_{(\mathbf{s},y)},B_{(\mathbf{s},y)}$ up to a monomial $\Delta$, and the details will be given below.

\item[Part c] delete $\Delta$ from $H$ to obtain $\gcd(A_{(\mathbf{s},y)},B_{(\mathbf{s},y)})$ and then $\gcd(A,B)$.
\end{description}

We will explain Part $\mathbf{b}$ in more details.
For matching different evaluations, we need to know the leading coefficient (or some other coefficients in a fixed degree). But it is hard to know in advance before we know the exact form of $G$. Luckily, in Part $\mathbf{a}$, we have found a vector $\mathbf{s}$ such that $G_{(\mathbf{s},y)}$ has a maximum isolated term, so the leading coefficient of $G$ w.r.t $y$ is a factor of $(x_1x_2\cdots x_n)^d$, where $d$ is the partial degree bound of $A,B$. So we regard the leading coefficient of $\gcd(A_{(\mathbf{s},y)},B_{(\mathbf{s},y)})$ as $(x_1x_2\cdots x_n)^d$. The result is only different from $\gcd(A_{(\mathbf{s},y)},B_{(\mathbf{s},y)})$ by a monomial factor which can be removed easily.

Denote $F_1=A_{(\mathbf{s},y)}$ and $F_2=B_{(\mathbf{s},y)}$.  Assume $\overrightarrow{\alpha}$ is a good point for $F_1,F_2\in \F_q[\X,y]$. Regard $C=\gcd(F_1,F_2)$ as the polynomial in $y$ with coefficients in $\F_q[\X]$ and assume
$$C=C_\ell y^{e_\ell}+\cdots+C_1y^{e_1},$$
where $C_i\in\F_q[\X]$ and $e_\ell>e_{\ell-1}>\cdots>e_1$. As $C$ has a maximum isolated term w.r.t $y$, $C_\ell$ is a monomial.
The image has the form
$$\gcd(F_1(y,\overrightarrow{\alpha}),F_2(y,\overrightarrow{\alpha}))=y^{e_\ell}+\frac{C_{\ell-1}(\overrightarrow{\alpha})}{C_\ell(\overrightarrow{\alpha})}y^{e_{\ell-1}}+\cdots+\frac{C_1(\overrightarrow{\alpha})}{C_\ell(\overrightarrow{\alpha})}y^{e_1}.$$
Let
\begin{eqnarray}\label{eq-3}
&&H=\Delta\cdot C =
 H_\ell y^{e_\ell}+H_{\ell-1}y^{e_{\ell-1}}+\cdots+H_1y^{e_1},
%
\end{eqnarray}
where $\Delta=\frac{(x_1\cdots x_n)^d}{C_\ell}$,
$H_i=\frac{(x_1\cdots x_n)^d}{C_\ell}C_i$, and in particular $H_\ell=(x_1\cdots x_n)^d$.
As $d$ is a partial degree bound of $A,B$, $C_\ell$ divides $(x_1\cdots x_n)^d$ and $\Delta$ is a monomial.
%
So we have
\begin{eqnarray}\label{eq-H3}
H(y,\overrightarrow{\alpha})=(\alpha_1\cdots \alpha_n)^d\cdot\gcd(F_1(y,\overrightarrow{\alpha}),F_2(y,\overrightarrow{\alpha})).
\end{eqnarray}
We can evaluate $H_i\in\F_q[\X],i=1,2,\dots,\ell-1$ at the point $\overrightarrow{\alpha}$. Varying $\overrightarrow{\alpha}$, we can recover all $H_i$'s from the evaluations by interpolation.

Part $\mathbf{b}$ can be  divided mainly into four steps.
\begin{itemize}
\item compute all term bounds $T_i$'s of all coefficients of $\gcd(A_{(\mathbf{s},y)},B_{(\mathbf{s},y)})$ w.r.t. $y$ using the technique of early termination in section \ref{sec-et};

\item  diversify all the coefficients of $\gcd(A_{(\mathbf{s},y)},B_{(\mathbf{s},y)})$ w.r.t. $y$
using the method given in section \ref{sec-dp} and
find a good point $\overrightarrow{\alpha}$ using the method given in section \ref{sec-gp};
\item  evaluate all the coefficients of $H(y,\overrightarrow{\alpha}^i),H(y,\overrightarrow{\alpha}_k^i)$ in \eqref{eq-H3},
which is possible because $\gcd(F_1(y,\overrightarrow{\alpha}^i),F_2(y,\overrightarrow{\alpha}^i)),\gcd(F_1(y,\overrightarrow{\alpha}_k^i),F_2(y,\overrightarrow{\alpha}_k^i))$ are monic;
\item interpolate all the coefficients of $H=\Delta\cdot\gcd(A_{(\mathbf{s},y)},B_{(\mathbf{s},y)})$ in \eqref{eq-3} from the evaluations by sparse polynomial interpolation.
\end{itemize}
\subsection{Primitive GCD algorithm}
\label{sec-a2}

We give a GCD algorithm in $\F_q[\X]$ for primitive polynomials.

\begin{algorithm}\label{alg-3}
Primitive GCD over finite fields

{\noindent\bf Input:}
\begin{itemize}
\item Two monomial primitive polynomials $A,B\in \F_q[\X]$.
\item A primitive element $\omega$ of $\F_q$.
\item A tolerance $\epsilon$.
\end{itemize}

{\noindent\bf Output:} $G=\gcd(A,B)$ with probability $\geq 1-\epsilon$; or ``Failure."

{\noindent Initial}

\begin{description}
\item[Step 0] Let $d=\max\{\deg_{x_i}A,\deg_{x_i}B,i=1,2,\dots,n\}$ and $D=\max\{\deg A,$ $\deg B\}$.
\item[Step 1] If $q<D$, then find an irreducible polynomial $\Upsilon(z)$ over $\F_q[z]$ of degree $k\geq \frac{\log D}{\log q}$. Construct finite field $\F_{q^k}$ as $\F_q[z]/(\Upsilon)$.
    For the convenience of description, in the following, we still denote $\F_{q^k}$ as $\F_q$. Find a primitive root of $\F_{q^k}$ and still denote it $\omega$.
\end{description}

{\noindent\bf Stage \uppercase\expandafter{\romannumeral1}: Find a vector $\mathbf{s}$ such that at least one of $A_{(\mathbf{s},y)},B_{(\mathbf{s},y)}$ has a maximum isolated term.}
\begin{description}
\item[Step 2] Let $N=2\min\{T_A-1,T_B-1\}$. Randomly choose $\mathbf{s}\in [1,N]^n$.  If both of $A_{(\mathbf{s},y)},B_{(\mathbf{s},y)}$ do not have a maximum isolated term, then repeat Step 2.
\item[Step 3] Set $F_1:=A_{(\mathbf{s},y)},F_2:=B_{(\mathbf{s},y)}$.
\end{description}

{\noindent\bf Stage \uppercase\expandafter{\romannumeral2}:
Find terms bound for all coefficients of $\Delta\cdot\gcd(A_{(\mathbf{s},y)},B_{(\mathbf{s},y)})$.}

\begin{description}
\item[Step 4] Find an irreducible polynomial $\Phi(z)$ over $\F_q[z]$ of degree $$r= \lceil\frac{\log\frac{1}{\varepsilon}+\log 86+2n\log (d+1)+2\log (nd)+\log \|\mathbf{s}\|_{\infty}}{\log q}\rceil.$$ Construct finite field $\F_{q^r}$ as $\F_q[z]/(\Phi)$.
\item[Step 5] Set $T:=1$. Randomly choose $\overrightarrow{\sigma}=(\sigma_1,\dots,\sigma_n)\in \F_{q^r}^{*n}$.
\item[Loop]
\item[Step 6]
For $i=T,T+1,\dots,2T-1$
\begin{description}
\item[a]If one of $\LC_y(F_1)(\overrightarrow{\sigma}^{i})$, $\LC_y(F_2)(\overrightarrow{\sigma}^{i})$ is zero, then return ``Failure."

\item[b]Compute the monic GCD of $F_1(y,\overrightarrow{\sigma}^{i})$ and $F_2(y,\overrightarrow{\sigma}^{i})$. Let
 $$\eta'_i:=\gcd(F_1(y,\overrightarrow{\sigma}^i),F_2(y,\overrightarrow{\sigma}^i)).$$
If one of $\eta'_i$ has different degree with others, then return ``Failure."
    \end{description}

\item[Step 7] Multiply $\eta'_i$ by the leading coefficient $(\sigma_1\cdots \sigma_n)^{i\cdot d}$.

For $i=T,T+1,\dots,2T-1$ do
$$\eta_{i}:=\eta'_i\cdot (\sigma_1\cdots \sigma_n)^{i\cdot d}.$$
Assume
    $$\eta_i:=c_{i,1}y^{e_1}+\cdots+c_{i,\ell}y^{e_\ell},i=1,\dots,2T-1.$$

\item[Step 8]
Construct Hankel matrices $\mathbf{H}_k:=(c_{i+j-1,k})_{i,j=1,\dots,T},k=1,2,\dots,\ell-1$.
If one of $\det(\mathbf{H}_j)$ is not zero for $j=1,2,\dots,\ell-1$, then $T:=2T$, and goto {\bf Loop}.
Write down $T_i$ for each $y^{e_i}$, which is the first $T$ for $\det(\mathbf{H}_i)=0$.

\item[Step 9] Let $T:=T-1$; and $T_i=T_i-1,i=1,2,\dots,\ell-1$.
\end{description}

{\noindent\bf Stage \uppercase\expandafter{\romannumeral3}: Choose good evaluation points and diversify the GCD.}
\begin{description}
\item[Step 10] Find an irreducible polynomial $\Phi'(z)$ over $\F_q[z]$ of degree $$m\geq \lceil\frac{\log\frac{1}{\varepsilon}+\log 42+\log (n+1)+2\log (ndT)}{\log q}\rceil.$$ Construct finite field $\F_{q^m}$ as $\F_q[z]/(\Phi')$.
\item[Step 11]
Randomly choose $\overrightarrow{\zeta}=(\zeta_1,\dots,\zeta_n)\in \F_{q^m}^{*n}$ and $\overrightarrow{\alpha}=(\alpha_1,\dots,\alpha_n)\in \F_{q^m}^{*n}$.
/* It will be proved that $H$ is diverse w.r.t. $y$ by $\overrightarrow{\zeta}$
and $\overrightarrow{\alpha}$ is a good point.*/

\item[Step 12]Compute $\widetilde{A}=A(\zeta_1x_1,\dots,\zeta_nx_n)$, $\widetilde{B}=B(\zeta_1x_1,\dots,\zeta_nx_n)$.
Compute $\widetilde{F}_1=(\widetilde{A})_{(\mathbf{s},y)}$, $\widetilde{F}_2=(\widetilde{B})_{(\mathbf{s},y)}$.
\end{description}

{\noindent\bf Stage \uppercase\expandafter{\romannumeral4}: Evaluate the GCD.}
\begin{description}
\item[Step 13]
For $i=1,2,\dots,2T$
\begin{description}
\item[a]If one of $\LC_y(\widetilde{F}_1)(\overrightarrow{\alpha}^{i})$, $\LC_y(\widetilde{F}_2)(\overrightarrow{\alpha}^{i})$, $\LC_y(\widetilde{F}_1)(\overrightarrow{\alpha}_k^{i})$, and $\LC_y(\widetilde{F}_2)(\overrightarrow{\alpha}_k^{i}),k=1,2,\dots,n$ is zero, then return ``Failure."

\item[b]Compute the monic univariate GCD of $\widetilde{F}_1(y,\overrightarrow{\alpha}^{i})$ and $\widetilde{F}_2(y,\overrightarrow{\alpha}^{i})$.
 $$f'_i:=\gcd(\widetilde{F}_1(y,\overrightarrow{\alpha}^i),\widetilde{F}_2
    (y,\overrightarrow{\alpha}^i)).$$
\item[c]Compute the monic univariate GCD of $\widetilde{F}_1(y,\overrightarrow{\alpha}_k^{i})$ and $\widetilde{F}_2(y,\overrightarrow{\alpha}_k^{i})$ for $k=1,2,\dots,n$ $$g'_{i,k}:=\gcd(\widetilde{F}_1(y,\overrightarrow{\alpha}_k^{i}),\widetilde{F}_2(y,\overrightarrow{\alpha}_k^{i})).$$
    \end{description}
 \end{description}

\begin{description}
\item[Step 14] Multiply $f'_i,g'_{i,k}$ by the leading coefficients $(\alpha_1\cdots \alpha_n)^{i\cdot d}$ and $(\alpha_1\cdots \alpha_n\cdot\omega)^{i\cdot d}$.

For $i=1,2,\dots,2T$ do
$$f_{i}:=f'_i\cdot (\alpha_1\cdots \alpha_n)^{i\cdot d}=c_{i,1}y^{e_1}+\cdots+c_{i,\ell}y^{e_\ell}$$
$$g_{i,k}:=g'_{i,k}\cdot (\alpha_1\cdots \alpha_n\cdot\omega)^{i\cdot d}=r_{i,k,1}y^{e_1}+\cdots+r_{i,k,\ell}y^{e_\ell}$$
for all $k=1,2,\dots,n$.
\end{description}

{\noindent\bf Stage \uppercase\expandafter{\romannumeral5}: Compute GCD by interpolation.}
\begin{description}
\item[Step 15]
For $j=1,2,\dots,\ell-1$,
compute the polynomials by Algorithm \ref{alg-2}:

$H'_j:={\rm Interpolation}(\omega,c_{i,j},r_{i,k,j},i=1,\dots,2T_i,k=1,2,\dots,n)$.

Set $H_j:=H'_j(\zeta_1^{-1}x_1,\dots,\zeta_n^{-1}x_n)$ and $H_\ell:=(x_1\cdots x_n)^d$.
\end{description}

{\noindent\bf Stage \uppercase\expandafter{\romannumeral6}: Compute the monomial primitive part.}
\begin{description}
\item[Step 16] For $i=1,\dots,n$, let $k_i = \min\{\deg(H_j,x_i),j=1,\dots,\ell\}$
  \newline /* $x_1^{k_1}\cdots x_n^{k_n}$ is the monomial content $\Delta$ of $\sum_{i=1}^{\ell}H_i$. $*/$.

\item[Step 17] Return the primitive part $(\sum_{i=1}^{\ell}H_i)/(x_1^{k_1}\cdots x_n^{k_n})$.
\end{description}

\end{algorithm}

\begin{theorem}\label{the-7}
Let $A,B\in \F_q[\X]$ be monomial  primitive polynomials and $\omega\in\F_q$ a fixed primitive root. Then Algorithm \ref{alg-3} is correct.
\begin{enumerate}
\item With probability $\geq 1-\varepsilon$, it returns the correct GCD and the complexity is $O^\sim(nDT_G(T_A+T_B)\log^2\frac{1}{\varepsilon} \log^2 q)$ bit operations.
\item The expected complexity is $O^\sim(n^3DT_G(T_A+T_B) \log^2 q)$ bit operations.
\end{enumerate}
\end{theorem}
\begin{proof}
The proof  is given in Section \ref{sec-5-2}.
\end{proof}
In Stage \uppercase\expandafter{\romannumeral1}, we find a suitable vector $\mathbf{s}$ for $A,B$ such that at least one of $A_{(\mathbf{s},y)},B_{(\mathbf{s},y)}$ has a maximum isolated term.
In Stage \uppercase\expandafter{\romannumeral2}, we compute the terms bound $T_i$ of coefficients of $\Delta\cdot\gcd(A_{(\mathbf{s},y)},B_{(\mathbf{s},y)})$ in $y$ by using the technique of early termination.
In Stage \uppercase\expandafter{\romannumeral3}, we diversify all the coefficients of $\Delta\cdot\gcd(A_{(\mathbf{s},y)},B_{(\mathbf{s},y)})$ w.r.t. $y$ by using  $\overrightarrow{\zeta}=(\zeta_1,\dots,\zeta_n)$ and choose the good evaluation point $\overrightarrow{\alpha}$.
In Stage \uppercase\expandafter{\romannumeral4}, we evaluate all the coefficients of $\Delta\cdot\gcd(A_{(\mathbf{s},y)},B_{(\mathbf{s},y)})$ at points $\overrightarrow{\alpha}^i,i=1,2,\dots,2T$ and $\overrightarrow{\alpha}_k^i,i=1,\dots,2T,k=1,\dots,n$.
In Stage \uppercase\expandafter{\romannumeral5}, we interpolate all the coefficients of $\Delta\cdot\gcd(A_{(\mathbf{s},y)},B_{(\mathbf{s},y)})$ at points $\overrightarrow{\alpha}^i,i=1,2,\dots,2T$ and $\overrightarrow{\alpha}_k^i,i=1,\dots,2T,k=1,\dots,n$ by sparse polynomial interpolation.
In Stage \uppercase\expandafter{\romannumeral6}, we remove the factor $\Delta$ from $\Delta\cdot\gcd(A_{(\mathbf{s},y)},B_{(\mathbf{s},y)})$ by computing the monomial content and return the GCD of $A,B$.

\subsection{GCD algorithm for polynomials over finite fields}
Based on Algorithm \ref{alg-3}, we give the complete  GCD algorithm  polynomials over finite fields.

\begin{algorithm}\label{alg-4}
GCD over finite fields

{\noindent\bf Input:}
\begin{itemize}
\item $A,B\in \F_q[\X]$.
\item A primitive element $\omega$ of $\F_q$.
\item A tolerance $\epsilon$.
\end{itemize}

{\noindent\bf Output:} $G=\gcd(A,B)$ with probability $\geq 1-\epsilon$; or ``Failure."

\begin{description}
\item[Step 1] Compute the monomial contents and the primitive parts of $A$ and $B$, and denote them by $C_A:=\MoCont(A),C_B:=\MoCont(B),P_A:=\MoPrim(A),P_B:=\MoPrim(B)$.
\item[Step 2] Compute the GCD $G'=\gcd(P_A,P_B)$ by Algorithm \ref{alg-3} with tolerance $\epsilon$.

\item[Step 3] Compute the GCD $C'=\gcd(C_A,C_B)$.
\item[Step 4] Return $G'\cdot C'$.
\end{description}

\end{algorithm}

\begin{theorem}
Let $A,B\in \F_q[\X]$ and $\omega\in\F_q$ a fixed primitive root. Then Algorithm \ref{alg-4} is correct.
\begin{enumerate}
\item With probability $\geq 1-\varepsilon$, it returns the correct GCD $G=\gcd(A,B)$ and the complexity is $O^\sim(nDT_G(T_A+T_B)\log^2\frac{1}{\varepsilon} \log^2 q)$ bit operations.
\item The expected complexity is $O^\sim(n^3DT_G(T_A+T_B) \log^2 q)$ bit operations.
\end{enumerate}
\end{theorem}
\begin{proof}
The correctness comes from the Equ. (\ref{eq-2}). Once $G'$ is computed correctly in Step 2, Algorithm \ref{alg-4} returns the correct polynomial. According to Theorem \ref{the-7}, we compute the correct $\gcd(P_A,P_B)$ with probability $\geq 1-\varepsilon$. So the correctness is proved.

Now we analyse the complexity.
In Step 1, since $\MoCont(A)$ is a monomial, to compute $\MoCont(A)$, it suffices to find each exponent of $x_i$,  which is equivalents to finding the minimum degrees of $x_i$'s in $A$. So the cost is $O(nT_A\log D)$ bit operations. As $\MoPrim(A)=A/\MoCont(A)$, to compute $\MoPrim(A)$, just subtract the exponents of $\MoCont(A)$ from the exponents of each term of $A$, which costs $O(nT_A\log D)$ bit operations. Similarly, the cost of computing $\MoCont(B)$ and $\MoPrim(B)$ is $O(nT_B\log D)$.
So the total cost is $O(n(T_A+T_B)\log D)$ bit operations.
In Step 2, by Theorem \ref{the-7}, the complexity is $O^\sim(nDT_G(T_A+T_B)\log^2\frac{1}{\varepsilon} \log^2 q)$ bit operations.
In Step 3, computing the GCD of two monomials is equivalent to comparing the exponent of each $x_i$ of $C_A$ and $C_B$, and the smaller one is the exponent of the GCD about $x_i$. So the cost is $O(n\log D)$ bit operations.
In Step 4, to compute the product, we can directly add the exponents of $C'$ to the exponents of all terms of $G'$, which requires $O(nT_G\log D)$ bit operations.

\end{proof}

\section{Experimental results}\label{sec-exp}
In this section, the practical performance  of our GCD
algorithm for polynomials over finite fields are given.
We compare with the default implementation of the GCD algorithm in Maple 2018.
%
The data are collected on a desktop with Windows system,
2.50GHz Core i5 processor and 8GB RAM memory.
The codes can be found in
https://github.com/huangqiaolong/Maple-Codes-GCD.
%

To test the average running times of the algorithm, we use the Maple command $randpoly$
to construct five pairs of random co-prime polynomials $A,B\in
\F_p[\X]$ and a polynomial $G$ within the given terms bound and degree bound, then expand $A\cdot G$ and $B\cdot G$ and compute $G=\gcd(A\cdot G,B\cdot G)$
with our algorithm and the Maple command ${\rm Gcd}(A,B) \mod p$.
The average times are collected. In our testing, we fix $p=10000019$ and use the primitive element $\omega = 6$. In our code, we do not use the expansion of finite fields.  A simple analysis shows that the success rate is $\geq 1-\frac{86n^2t^2d^2\min\{t_A,t_B\}+168(n+1)t^2n^2d^2\min\{t_A,t_B\}}{p}$,
where $t_A:=\#(A\cdot G)$, $t_B=\#(B\cdot G)$, $t=\#G$ and $d$ is the partial degree bound of $A\cdot G$ and $B\cdot G$.

Three benchmarks are used for the experiments and the results are given in Figures \ref{figrt}-\ref{figrd2},
where  the red lines are the timings of our algorithm  and the black lines are the timings of the Maple code.

{\bf Benchmark 1.}
For the first benchmark, we fix the degrees of $A,B,G$ for
$n=6$, $\deg A=\deg B=\deg G=30$, and change $\# A=\#B=\#G=T$ from $2$ to $152$.
The computing times are shown in Figure \ref{figrt},
where we take 60 seconds as the threshold: once exceeding 60 seconds, we terminate the computing.
From this figure, we can see that
the Maple code can compute GCDs with terms up to 18
and our code can compute GCDs with terms up to 150.
%

{\bf Benchmark 2.}
For the second benchmark, we fix the terms of $A,B,G$ for
$\# A=\#B=\#G=30$, and the degrees of $A,B,G$ for
$\deg A=\deg B=\deg G=100$, and change $n$ from $1$ to $200$.
The computing times are shown in Figure \ref{figrn},
where the threshold is set to be  60 seconds.
From this figure, we can see that
the Maple code can compute GCDs with numbers of variables up to 3
and our code can compute GCDs with numbers of variables up to 200,
so the new algorithm has about 2-orders of magnitude improvement.
The computing time of the Maple code increases rapidly when $n>3$
and is $\ge 600s$ for $n\ge10$.

{\bf Benchmark 3.}
For the third benchmark, we fix the terms of $A,B,G$ for
$n=6$, $\# A=\#B=\#G=30$, and change $D=\deg(A)=\deg(B)=\deg(G)$ from $5$ to $29525$ in increments of $500$.
The computing times are shown in Figure \ref{figrd1},
where the threshold is set to be  100 seconds.
From this figure, we can see that
the Maple code can compute GCDs with degrees up to 23
and our code can compute GCDs with degrees up to 29525,
so the new algorithm has about 3-orders of magnitude improvement.
In Figure \ref{figrd2}, we give  more details by changing $D=\deg(A)=\deg(B)=\deg(G)$ from $5$ to $14$. Figure \ref{figrd2} shows that $D=8$ is the intersection point. When $D<8$, the Maple code is better than ours, but when $D\geq 8$, our algorithm costs less and increase slowly until $d=29525$.
The timings for the Maple code increase drastically after $D\ge 23$.
%
%
These experimental results also validate the complexities in Table \ref{tab-1}.

\begin{figure}[!hptb]
\begin{minipage}[t]{0.50\linewidth}
\centering
\includegraphics[scale=0.18]{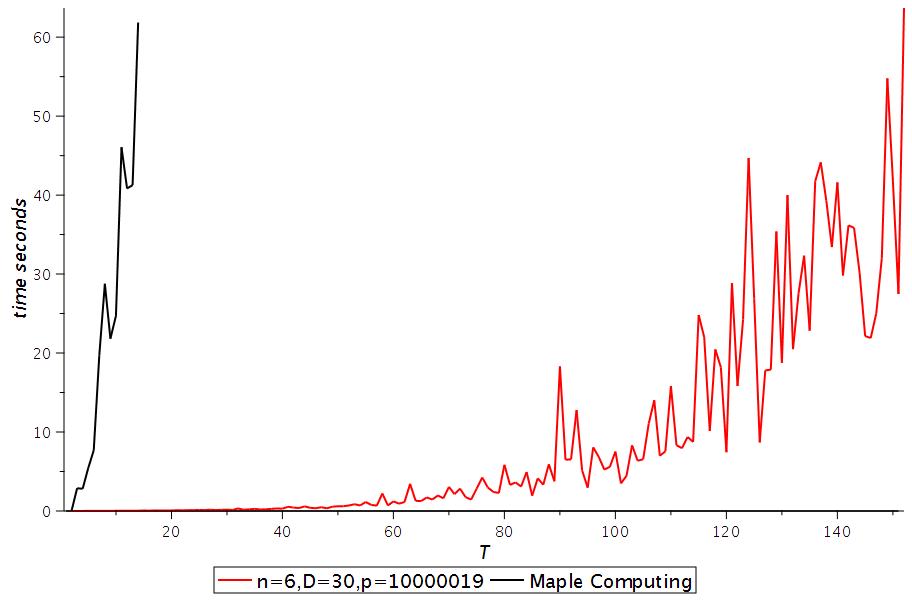}
\caption{Average running time  with varying terms}
\label{figrt}
\end{minipage}
\begin{minipage}[t]{0.44\linewidth}
\centering
\includegraphics[scale=0.18]{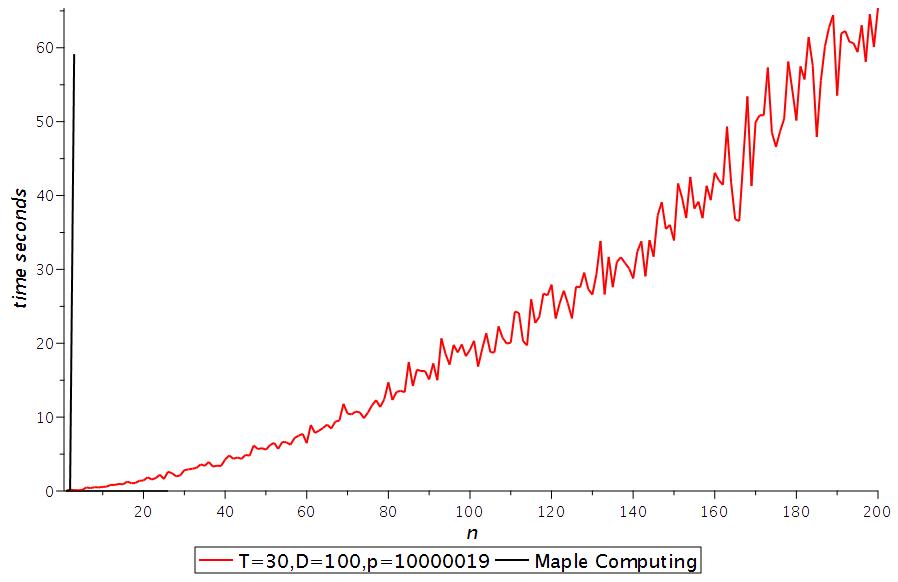}
\caption{Average running time with varying number of variables}\label{figrn}
\end{minipage}
\end{figure}

\begin{figure}[!hptb]
\begin{minipage}[t]{0.48\linewidth}
\centering
\includegraphics[scale=0.18]{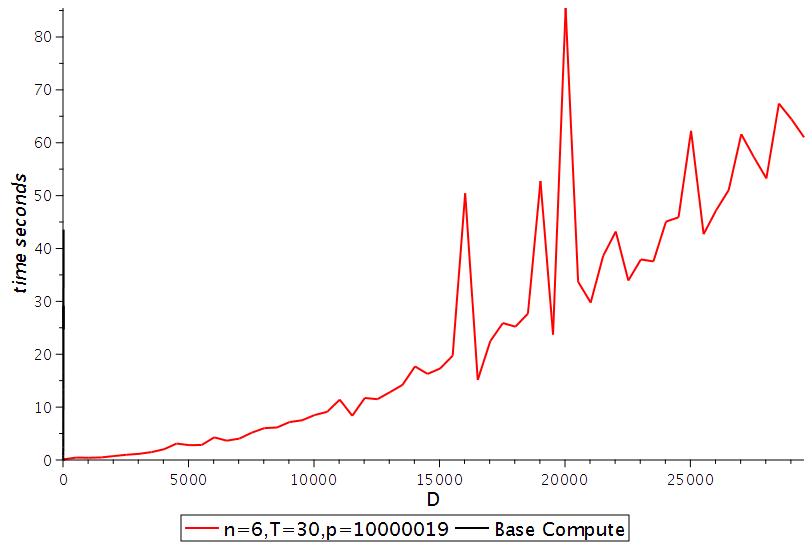}
\caption{Average running time with varying degree: the global picture for all degrees} \label{figrd1}
\end{minipage}
\begin{minipage}[t]{0.49\linewidth}
\centering
\includegraphics[scale=0.19]{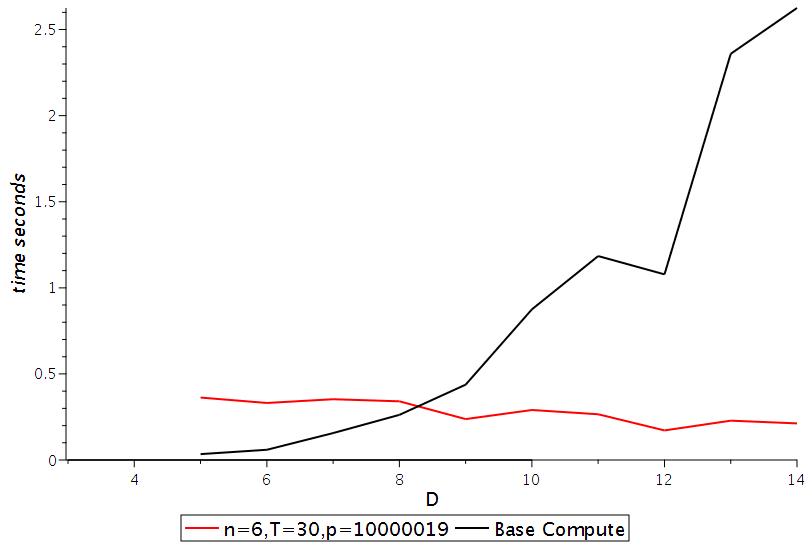}
\caption{Average running time with varying degree: the picture with degrees $\le 14$} \label{figrd2}
\end{minipage}
\end{figure}

\begin{remark}
In our Maple code, in Step 2 of Algorithm \ref{alg-3}, to isolate the maximum term of $A$ or $B$, we actually let $N=1,2,2^2,2^3,\dots$ because the success rate of isolation is very high in practice even if $N$ is small.
In Step 8, we actually let $T=1,2,3,\dots$ and increase $T$ by only $1$ each time, as the cost of computing the GCD of two univariate polynomials is more expensive than testing the Hankel matrices.
\end{remark}

\section{Proof of Theorem \ref{the-7}}\label{sec-5-2}
We first prove some lemmas.

\begin{lemma}\label{lm-5}
Let $A\in\F_q[\X]$,  $D=\deg A$, and $t=\#A$.
If $A$ contains all variables $x_i$'s, then $Dt\geq n$.
\end{lemma}
\begin{proof}
Assume $A=c_1M_1+\cdots+c_tM_t$ and each $M_i$ contains $k_i$ different variables.
Then $Dt\geq k_1+k_2+\cdots+k_t$. As $A$ contains all variables $x_i$'s, $k_1+k_2+\cdots+k_t\geq n$.
So we have $Dt\geq n$.
\end{proof}

Let $H=\Delta\cdot \gcd(A_{(\mathbf{s},y)},B_{(\mathbf{s},y)})$ defined in \eqref{eq-3} have the form $H=H_\ell y^{e_\ell}+H_{\ell-1}y^{e_{\ell-1}}+\cdots+H_1y^{e_1}$, where $H_i\in \F_q[\X]$, $e_1<e_2<\cdots<e_{\ell}$ and $\deg H_i\leq D+nd$.
\begin{lemma}\label{lm-3}
In Stage \uppercase\expandafter{\romannumeral2}, Algorithm \ref{alg-3} returns correct bounds $T_i$'s, which satisfy $\#H_i\leq T_i< 2\#H_i$, with probability $\geq 1-\frac{\varepsilon}{2}$.
\end{lemma}
\proof
We use $\overrightarrow{\sigma}$ to test the term bounds of all the coefficients of $H$ w.r.t $y$. Assume $t_i=\#H_i$ and $t=\max \{t_i,i=1,\dots,\ell\}$. Without loss of generality, consider $H_1$.
We test if $\lfloor \log t_1\rfloor+2$ determinants are zero, whose orders are $2^0,2^1,\dots,2^{\lfloor \log t_1\rfloor+1}$.
Assume the corresponding determinants are ${\bf DetH}_i,i=0,1,\dots,\lfloor \log t_1\rfloor+1$, by \cite{KaltofenL03}, $\deg{\bf DetH}_i \leq (2^i)^2 (D+nd)$.
The success of early termination is decided by the selection of $\overrightarrow{\sigma}$ such that all ${\bf DetH}_i(\overrightarrow{\sigma}) \neq 0,i=0,1,\dots,\lfloor \log t_1\rfloor$. Multiply all of them into one polynomial $\Gamma_1:=\prod_{i=0}^{\lfloor \log t_1\rfloor}{\bf DetH}_i$, with degree
$\leq \sum_{i=0}^{\lfloor \log t_1\rfloor}2^{2i}(D+nd)\leq \frac{4t_1^2}{3}(D+nd)$. For the same reason, for each $H_i$, there exists a non-zero condition polynomial $\Gamma_i$, such that if $\Gamma_i(\overrightarrow{\sigma})\neq 0$. Then Step 9 returns a correct term bound $T_i$ for $H_i$. The degree of $\Gamma_i$ $\leq \frac{4t_i^2}{3}(D+nd)$.

As the evaluations of $H_i$ come from the images of GCD of $F_1$ and $F_2$, $\overrightarrow{\sigma}^i,i=1,2\dots,4t$ should all be good points for $F_1$ and $F_2$.

Let $R=\res_{y}(F_1/\gcd(F_1,F_2),F_2/\gcd(F_1,F_2))$.
Set $L=\LC_y(F_1)\cdot \LC_y(F_2)\cdot R\in \F_q[\X]$.
As $L$ is a non-zero polynomial, $L(x_1^i,\dots,x_n^i)$ is also a non-zero polynomial.
So $(\sigma_1^i,\dots,\sigma_n^i)$ is a good point if $(\sigma_1,\dots,\sigma_n)$ is not a zero of $L(x^i_1,\dots,x^i_n)$.
Define a non-zero polynomial
 $$L_{\rm term}=\prod_{i=1}^{4t}L(x_1^i,\dots,x_n^i)\prod_{i=1}^{\ell-1}\Gamma_i.$$
So Step 9 returns term bounds $T_i$ satisfying $\#H_i\leq T_i< 2\#H_i$ if $\overrightarrow{\sigma}$ satisfies $L_{\rm term}(\overrightarrow{\sigma})\neq 0$.
By Lemma \ref{the-14}, $\deg R\leq 2\|\mathbf{s}\|_{\infty}\deg A\deg B\leq 2\|\mathbf{s}\|_{\infty}D^{2}$, so we have $\deg L\leq\deg(\LC_y(A))+\deg(\LC_y(B))+\deg R\leq 2D+2\|\mathbf{s}\|_{\infty}D^2$.
So $\deg L(x_1^i,\dots,x_n^i)\leq i(2D+2\|\mathbf{s}\|_{\infty}D^2)$, which implies
 $\deg \prod_{i=1}^{4t}L(x_1^i,\dots,x_n^i)\leq \sum_{i=1}^{4t} i(2D+2\|\mathbf{s}\|_{\infty}D^2)=2t(4t+1)(2D+2\|\mathbf{s}\|_{\infty}D^2)$.
As $D\leq nd$ and $\ell\leq \|\mathbf{s}\|_{\infty}D+1$, we have $\deg L_{\rm term}\leq 2t(4t+1)(2D+2\|\mathbf{s}\|_{\infty}D^2)+\frac{4t^2}{3}(D+nd)(\ell-1)\leq 2t(4t+1)(2D+2\|\mathbf{s}\|_{\infty}D^2)+\frac{4t^2}{3}(D+nd)\|\mathbf{s}\|_{\infty}D<43n^2t^2d^2\|\mathbf{s}\|_{\infty}$.
Since $t\leq (d+1)^n$, $\deg L_{\rm term}< 43(d+1)^{2n}n^2d^2\|\mathbf{s}\|_{\infty}$.
Since $r\geq \frac{\log\frac{1}{\varepsilon}+\log 86+2n\log (d+1)+2\log (nd)+\log \|\mathbf{s}\|_{\infty}}{\log q}$, $q^r\geq \frac{86}{\varepsilon}(d+1)^{2n}n^2d^2\|\mathbf{s}\|_{\infty}$. So by Lemma \ref{lm-6}, Stage \uppercase\expandafter{\romannumeral2} returns correct term bounds with probability
$$\geq 1-\frac{\deg L_{\rm term}}{q^r-1}\geq 1-\frac{43(d+1)^{2n}n^2d^2\|\mathbf{s}\|_{\infty}-1}{\frac{86}{\varepsilon}(d+1)^{2n}n^2d^2\|\mathbf{s}\|_{\infty}-1}\geq 1-\frac{\varepsilon}{2}.$$
\qed

\begin{lemma}\label{lm-4}
In Stages \uppercase\expandafter{\romannumeral3}, \uppercase\expandafter{\romannumeral4}, \uppercase\expandafter{\romannumeral5} and \uppercase\expandafter{\romannumeral6}, if $T_i\geq \#G_i$, Algorithm \ref{alg-3} returns the correct $G=\gcd(A,B)$, with probability $\geq 1-\frac{\varepsilon}{2}$.
\end{lemma}
\begin{proof}
Once $H=\sum_{i=1}^{\ell}H_iy^{e_i}$ is computed correctly in Stage \uppercase\expandafter{\romannumeral5}, we can obtain the correct $G=\gcd(A,B)$. So we analyse the probability of obtaining correct $H_j$'s in Step 15.
We will prove  that
$H$ is diverse w.r.t. $y$ by $\overrightarrow{\zeta}$
and $\overrightarrow{\alpha}$ is a good point.

In our algorithm, $H$ should be diverse w.r.t. $y$. So
 a point $\overrightarrow{\zeta}=(\zeta_1,\dots,\zeta_n)$ should be chosen, such that $H_1(\overrightarrow{\zeta}\X),\dots,
H_{\ell-1}(\overrightarrow{\zeta}\X)$ are all diverse.
Considering $H_1$ and assume that $H_1=\Delta\cdot(c_1M_1+\cdots+c_uM_u)$, $\overrightarrow{\zeta}$ diversifying $H_1$ means
$\prod_{i\neq j} (c_iM_i(\overrightarrow{\zeta})-c_jM_j(\overrightarrow{\zeta}))\neq 0$.
Set $U_1=\prod_{i\neq j} (c_iM_i-c_jM_j)$. Then  $\overrightarrow{\zeta}$ diversifies $H_1$ if it is not a zero of $U_1$. For the same reason, we can set polynomials $U_2$,\dots,$U_{\ell-1}$ for $H_2,\dots,H_{\ell-1}$ and set $U=U_1\cdots U_{\ell-1}$. If $U(\overrightarrow{\zeta})\neq 0$, then the point $\overrightarrow{\zeta}$ satisfies our diversification condition.
Estimate the degree bound of $U$,
$$\deg U=\deg U_1+\cdots+\deg U_{\ell-1}\leq \frac{T(T-1)}{2}D(\ell-1)\leq \frac{T(T-1)}{2}D^2\|\mathbf{s}\|_{\infty}.$$

Now we evaluate the GCD in Step 13.
Consider point $(\alpha_1,\dots,\alpha_n)$. As $F_1$ and $F_2$ are diversified, consider the polynomial  $R_{\overrightarrow{\zeta}}=\res_{y}(\widetilde{F}_1/\gcd(\widetilde{F}_1,\widetilde{F}_2),\widetilde{F}_2/\gcd(\widetilde{F}_1,\widetilde{F}_2))$.
Set $Q=\LC_y(\widetilde{F}_1)\cdot \LC_y(\widetilde{F}_2)\cdot R_{\overrightarrow{\zeta}}\in \F_{q^m}[\X]$.
As $Q$ is a non-zero polynomial, $Q(x_1^i,\dots,x_n^i)$ is also a non-zero polynomial. So $(\alpha_1^i,\dots,\alpha_n^i)$ is a good point if $(\alpha_1,\dots,\alpha_n)$ is not a zero of $Q(x^i_1,\dots,x^i_n)$.
Clearly, $Q(x^i_1,\dots,\omega^i x^i_k,\dots,x_n^i)$ is a non-zero polynomial. For the same reason, $(\alpha^i_1,\dots,(\alpha_k\omega)^i,\dots,\alpha_n^i)$ is a good point if $(\alpha_1,\dots,\alpha_n)$ is not a zero of $Q(x^i_1,\dots,\omega^i x^i_k,\dots,x_n^i)$.
Define a non-zero polynomial $$Q_{\rm good}=\prod_{i=1}^{2T}Q(x_1^i,\dots,x_n^i)
\prod_{k=1}^n\prod_{j=1}^{2T}Q(x^j_1,\dots,x^j_k\omega^j,\dots,x^j_n).$$
So $(\alpha_1,\dots,\alpha_n)$ is a $2T$-$\omega$ good point for $\widetilde{F}_1,\widetilde{F}_2$ if and only if
$$Q_{\rm good}(\overrightarrow{\alpha},\overrightarrow{\zeta})\neq 0.$$
If $\overrightarrow{\zeta}$ are constants, then $\deg R_{\overrightarrow{\zeta}}\leq 2\|\mathbf{s}\|_{\infty}\deg A\deg B\leq 2\|\mathbf{s}\|_{\infty}D^{2}$. Now regard $\overrightarrow{\zeta}$ as variables. So $\deg_{\X} R_{\overrightarrow{\zeta}}\leq 2\|\mathbf{s}\|_{\infty}\deg_{\X} A\deg_{\X} B\leq 2\|\mathbf{s}\|_{\infty}D^{2}$ and $\deg_{\overrightarrow{\zeta}} R_{\overrightarrow{\zeta}}\leq 2\|\mathbf{s}\|_{\infty}\deg_{\overrightarrow{\zeta}} A\deg_{\overrightarrow{\zeta}} B\leq 2\|\mathbf{s}\|_{\infty}D^{2}$.
We thus have $\deg_{\X} Q\leq\deg_{\X}(\LC_y(\widetilde{F}_1))+\deg_{\X} R+\deg_{\X}(\LC_y(\widetilde{F}_2))\leq 2D+2\|\mathbf{s}\|_{\infty}D^2$ and $\deg_{\overrightarrow{\zeta}} Q\leq\deg_{\overrightarrow{\zeta}}(\LC_y(\widetilde{F}_1))+\deg_{\overrightarrow{\zeta}}(\LC_y(\widetilde{F}_2))+\deg_{\overrightarrow{\zeta}} R_{\overrightarrow{\zeta}}\leq 2D+2\|\mathbf{s}\|_{\infty}D^2$.
Then $\deg Q(x_1^i,\dots,x_n^i)\leq (2D+2\|\mathbf{s}\|_{\infty}D^2)+i\cdot (2D+2\|\mathbf{s}\|_{\infty}D^2)=(2D+2\|\mathbf{s}\|_{\infty}D^2)(i+1)$, which implies
\newline
 $\deg \prod_{i=1}^{2T}Q(x_1^i,\dots,x_n^i)\leq \sum_{i=1}^{2T} (i+1)(2D+2\|\mathbf{s}\|_{\infty}D^2)=T(2T+3)(2D+2\|\mathbf{s}\|_{\infty}D^2)$.
So $\deg Q_{\rm good}\leq (n+1)T(2T+3)(2D+2\|\mathbf{s}\|_{\infty}D^2)$.

Now we turn to Stage \uppercase\expandafter{\romannumeral5}. The correctness of $H_i$ comes from the correctness of the interpolation.
According to Assumption \ref{ass-2}, as $\deg H_i\leq D+nd$, $\overrightarrow{\alpha}$ should not vanish a polynomial with degree $\leq (n+1)\frac{T(T-1)}{2}(D+nd)(\ell-1)\leq (n+1)\frac{T(T-1)}{2}(D+nd)(\|\mathbf{s}\|_{\infty}D)$.
So the total degree of the three condition polynomials is $< 21(n+1)T^2n^2d^2\|\mathbf{s}\|_{\infty}$. As $m\geq \frac{\log\frac{1}{\varepsilon}+\log 42+\log (n+1)+2\log (ndT)}{\log q}$, $q^m\geq \frac{42}{\varepsilon}(n+1)T^2n^2d^2\|\mathbf{s}\|_{\infty}$.
By Lemma \ref{lm-6}, if $T$ is the upper bound of all $\#H_i$'s and the interpolation computes the correct polynomials with probability
$$\geq1-\frac{21(n+1)T^2n^2d^2\|\mathbf{s}\|_{\infty}-1}{q^m-1}\geq 1-\frac{\varepsilon}{2}.$$
\end{proof}

We now prove (1) of Theorem \ref{the-7}.
\begin{proof}
By Lemma \ref{lm-3}, $T_i$'s are upper bounds of $H_i$'s with probability $\geq 1-\frac{\varepsilon}{2}$.
By Lemma \ref{lm-4}, if each $T_i$ is an upper bound for $\#H_i$, then the interpolation algorithm computes the correct polynomials with probability $1-\frac{\varepsilon}{2}$.
So totally, Algorithm \ref{alg-3} returns the correct polynomials with probability $(1-\frac{\varepsilon}{2})^2\geq 1-\varepsilon$.
The correctness is proved.

We analyse the complexity. Here $T$ is the upper bound of $\#H_i$'s and $T_G=\sum_{i=1}^\ell \#H_i=\#G$

{\bf Stage} \uppercase\expandafter{\romannumeral1}: In Step 2, randomly choosing a vector $\mathbf{s}$ costs $O(n\log N)$ bit operations. In Step 3, computing $A_{(\mathbf{s},y)},B_{(\mathbf{s},y)}$ costs $O^\sim(n(T_A+T_B)(\log d+\log N))$ bit operations.
By Theorem \ref{the-4}, $\mathbf{s}$ is a suitable vector for $A$ or $B$ with probability $\geq \frac12$, so the expected cost is $O^\sim(n(T_A+T_B)(\log d+\log N))$ bit operations. Since $N\in O(\min\{T_A,T_B\})$, the expected cost is $O^\sim(n(T_A+T_B)\log d)$ bit operations

{\bf Stage} \uppercase\expandafter{\romannumeral2}:
In Step 6, we compute $F_1(y,\overrightarrow{\sigma}^{i})$ and $F_2(y,\overrightarrow{\sigma}^{i})$. As the partial degree of $F_1(y,\overrightarrow{\sigma}^i)$ is $O(d\|\mathbf{s}\|_\infty)$, the complexity is $O^\sim(nT_A\log d\log q^r+TT_A\log q^r)$ bit operations. Plus the cost for $F_2(y,\overrightarrow{\sigma})$, the total complexity is $O^\sim(n(T_A+T_B)\log d\log q^r+T(T_A+T_B)\log q^r)$ bit operations.
To compute the GCD of $F_1(y,\overrightarrow{\sigma}^i)$ and $F_2(y,\overrightarrow{\sigma}^i)$, the complexity is $O^\sim(TD\|\mathbf{s}\|_{\infty}\log q^r)$ bit operations.
In Step 8, to test the Hankel matrices, it costs $O^\sim(T\log q^r)$ bit operations.
So the complexity is $O^\sim(nT(T_A+T_B)\log d\log q^r+T(T_A+T_B)\log q^r+TD\min\{T_A,T_B\}\log q^r)$ bit operations.

Since $q^r$ is $O(\frac{1}{\varepsilon}d^{2n}n^2d^2\|\mathbf{s}\|_{\infty})$, the complexity is 
$O^\sim(n(T_A+T_B)\log D(n\log d+\log \frac{1}{\varepsilon})\log q+T(T_A+T_B)(n\log d+\log \frac{1}{\varepsilon})\log q+TD\min\{T_A,T_B\}(n+\log \frac{1}{\varepsilon})\log q)$ bit operations.

{\bf Stage} \uppercase\expandafter{\romannumeral3}: In Step 12, the cost is $O^\sim(n(T_A+T_B)\log d\log q^m)$ bit operations.

{\bf Stage} \uppercase\expandafter{\romannumeral4}: In Step 13, we compute $\widetilde{F}_1(y,\overrightarrow{\alpha}^{i})$, $\widetilde{F}_2(y,\overrightarrow{\alpha}^{i})$, $\widetilde{F}_1(y,\overrightarrow{\alpha}_k^{i})$ and $\widetilde{F}_2(y,\overrightarrow{\alpha}_k^{i})$. The complexity is $O^\sim(nT(T_A+T_B)\log d\log q^m)$ bit operations, which is $O^\sim(nT(T_A+T_B)(\log\frac{1}{\varepsilon}+\log d)\log d\log q)$ bit operations.
To compute the GCDs, the complexity is $O^\sim(nTD\|\mathbf{s}\|_{\infty}\log q^m)$ bit operations, which is $O^\sim(nTD\min\{T_A,T_B\}\log\frac{1}{\varepsilon}\log q)$ bit operations.

{\bf Stage} \uppercase\expandafter{\romannumeral5}:
In Step 15, for each interpolation of $H_i$, as $\deg_{x_j} H_i\leq2d$ for any $x_j,j=1,\dots,n$ and $\#H_i\leq T_i$,  by Theorem \ref{the-2}, the cost is $O^\sim(nT_i\log^2q^m+nT_i\sqrt{d}\log q)$ bit operations, which is $O^\sim(nT_i\log^2 d\log^2\frac{1}{\varepsilon}\log^2q+nT_i\sqrt{d}\log q)$ bit operations. So the total complexity is $O^\sim(nT_G\log^2 d\log^2\frac{1}{\varepsilon}\log^2q+nT_G\sqrt{d}\log q)$ bit operations.

Shoup \cite{shoup1994fast}  presented an algorithm to construct an irreducible polynomial of  degree $k$ over finite field $\F_q$ with an expected number of $O^\sim(k^2 +k \log q)$ operations in $\F_q$, which is $O^\sim(k^2\log q+k\log^2 q)$ bit operations.
So the complexity for constructing irreducible polynomials of degrees $m$ and $r$ is $O^\sim(n^2\log^2 (d\|\mathbf{s}\|_{\infty})\log^2\frac{1}{\varepsilon}\log q)$ bit operations.
Actually, by Lemma \ref{lm-5}, $nD(T_A+T_B)\geq n^2$. So the total complexity of our algorithm  is $O^\sim(nDT_G(T_A+T_B)\log^2\frac{1}{\varepsilon} \log^2 q)$ bit operations.

As in Step 1, we always let $q>D$. If $q<D$, we extend $\F_{q}$ to $\F_{q'}$ with $q'>D$. Finding a new primitive root costs $O(q'^{\frac{1}{4}+\epsilon})=O(D^{\frac{1}{4}+\epsilon})$ bit operations. So if $q\leq D$, we use $q'$ instead of $q$, the complexity is $O^\sim(nDT_G(T_A+T_B)\log^2\frac{1}{\varepsilon} \log^2 q'+n^2\log^2 D\log^2\frac{1}{\varepsilon}\log q')$ bit operations.
As $\log q'=\log D$, the cost is $O^\sim(nDT_G(T_A+T_B)\log^2\frac{1}{\varepsilon})$ bit operations.
So the cost is in $O^\sim(nDT_G(T_A+T_B)\log^2\frac{1}{\varepsilon}\log^2 q)$ bit operations.
\end{proof}

We now prove (2) of Theorem \ref{the-7}.
\proof

 We consider the worst case. Double the terms bound $T$, as the bad points for $A_{(\mathbf{s},y)}$ and $B_{(\mathbf{s},y)}$, $T$ becomes $O(d^n)$, the complexity for wrong computing is at most $O^\sim(nDd^n(T_A+T_B)\log^2\frac{1}{\varepsilon} \log^2 q)$ bit operations. But it happens only with probability $\leq \varepsilon$.  So the expected complexity is $$O^\sim((1-\varepsilon)nDT_G(T_A+T_B)\log^2\frac{1}{\varepsilon} \log^2 q+\varepsilon nDd^n(T_A+T_B)\log^2\frac{1}{\varepsilon} \log^2 q)$$
Now we choose $\varepsilon=\frac{1}{nDd^n(T_A+T_B)}$. Then the expected complexity is $O^\sim(n^3DT_G(T_A+T_B) \log^2 q)$ bit operations.
\qed

\section{Conclusion}
In this paper, we proposed a new method for computing sparse GCDs of multivariate polynomials. Our algorithm works for polynomials over any finite field. We map the multivariate polynomials into univariate ones which keeps the sparse structure. Then recover the target multivariate GCD via a variant of Ben-Or/Tiwari's interpolation algorithm over finite field.
We also give the explicit bit complexity for the algorithm, which is better than
that of Zippel's algorithm.
The algorithm is shown to be 1-3 orders of magnitude faster than the default Maple GCD codes for various benchmarks.
%
%
%


\providecommand{\bysame}{\leavevmode\hbox to3em{\hrulefill}\thinspace}
\providecommand{\MR}{\relax\ifhmode\unskip\space\fi MR }
\providecommand{\MRhref}[2]{%
  \href{http://www.ams.org/mathscinet-getitem?mr=#1}{#2}
}
\providecommand{\href}[2]{#2}

\end{document}